\documentclass[journal,11pt,onecolumn]{IEEEtran}

\usepackage{setspace} 

\usepackage[usenames, dvipsnames]{color}
\usepackage{citesort}
\usepackage{pifont}
\ifCLASSINFOpdf
  \usepackage[pdftex]{graphicx}
	\usepackage{epstopdf}
\else
  \usepackage[dvips]{graphicx}
\fi
\usepackage[cmex10]{amsmath}
\usepackage{algorithmic}
\usepackage{amsmath}
\usepackage{amssymb}
\usepackage{array}
\usepackage{mdwmath}
\usepackage{mdwtab}
\usepackage{eqparbox}
\usepackage{multirow}
\usepackage{fixltx2e}
\usepackage{enumerate}
\usepackage{url}
\usepackage{tabularx}
\usepackage{verbatim}
\usepackage{color}
\usepackage{amsthm}
\usepackage{xcolor}
\usepackage[shortlabels]{enumitem}

\newtheorem{theorem}{Theorem}

\DeclareMathOperator*{\argmin}{argmin}

\newcommand{\x}{{\mathbf{x}}}
\newcommand{\y}{{\mathbf{y}}}
\newcommand{\z}{{\mathbf{z}}}
\newcommand{\A}{{\mathbf{A}}}
\newcommand{\btheta}{{\boldsymbol{\theta}}}

\begin{document}
\doublespacing
%
\title{Understanding Untrained Deep Models for Inverse Problems: Algorithms and Theory}
\author{Ismail~Alkhouri\footnote{The first three authors contributed equally. © 2025 IEEE. Personal use of this material is permitted. Permission from IEEE must be obtained for all other uses, in any current or future media, including reprinting/republishing this material for advertising or promotional purposes, creating new collective works, for resale or redistribution to servers or lists, or reuse of any copyrighted component of this work in other works.},~\IEEEmembership{Member,~IEEE},  Evan~Bell, 
Avrajit~Ghosh,
Shijun~Liang,~\IEEEmembership{Student Member,~IEEE}, 
Rongrong~Wang,~\IEEEmembership{Senior Member,~IEEE}, Saiprasad~Ravishankar,~\IEEEmembership{Senior Member,~IEEE} 
}

\maketitle
\vspace{-2cm}

\begin{abstract}

In recent years, deep learning methods have been extensively developed for inverse imaging problems (IIPs), encompassing supervised, self-supervised, and generative approaches. Most of these methods require large amounts of labeled or unlabeled training data to learn effective models. However, in many practical applications, such as medical image reconstruction, extensive training datasets are often unavailable or limited. A significant milestone in addressing this challenge came in 2018 with the work of Ulyanov et al., which introduced the Deep Image Prior (DIP)—the first training-data-free \textcolor{black}{convolutional} neural network method for IIPs. Unlike conventional deep learning approaches, DIP requires only a convolutional neural network, the noisy measurements, and a forward operator. By leveraging the implicit regularization of deep networks initialized with random noise, DIP can learn and restore image structures without relying on external datasets. However, a well-known limitation of DIP is its susceptibility to \textcolor{black}{over-fitting}, primarily due to the over-parameterization of the network. In this tutorial paper, we provide a comprehensive review of DIP, including a theoretical analysis of its training dynamics.
We also categorize and discuss recent advancements in DIP-based methods aimed at mitigating \textcolor{black}{over-fitting}, including techniques such as regularization, network re-parameterization, and early stopping. Furthermore, we discuss approaches that combine DIP with pre-trained neural networks, present empirical comparison results against data-centric methods, and highlight open research questions and future directions.
\end{abstract}

\vspace{-0.2in}

\vspace{-0.2in}
\section{Introduction}\label{sec: Intro}









Inverse imaging problems (IIPs) arise across a variety of real-world applications~\cite{8434321,8103129}. In these applications, the goal is to estimate an unknown signal $\mathbf{x}\in \mathbb{R}^n$ from its measurements or a degraded signal $\mathbf{y}\in \mathbb{R}^m$, which are often corrupted by noise. Mathematically, they are typically \textcolor{black}{related as}
$\mathbf{y} = \mathcal{A}(\mathbf{x}) + \mathbf{n}\:,$ 
where $\mathcal{A}(\cdot) : \mathbb{R}^n\rightarrow \mathbb{R}^m$ (with $m\leq n$) represents a linear or non-linear forward model capturing the measurement process, and $\mathbf{n}\in \mathbb{R}^m$ denotes the general noise present in the measurements. Exactly solving IIPs is often challenging due to their ill-posedness and is commonly formulated as the optimization problem 
\begin{equation}\label{eqn: main}
    \min_{\mathbf{x}} \ell(\mathcal{A}(\mathbf{x}),\mathbf{y}) + \lambda R(\mathbf{x})\:,
\end{equation}
where $\ell(\cdot\:,\cdot)$ is a data-fitting loss capturing the fidelity between the estimated signal and observed measurements, alongside regularizer $R(\cdot)$ with a non-negative weighting $\lambda$, representing the signal prior. \textcolor{black}{Classical methods for solving~\eqref{eqn: main} include frameworks such as Compressed Sensing (CS)~\cite{lustig2007sparse} where the key insight is although the estimated image may not be low-dimensional, it often has a sparse representation in some known basis. Therefore, }various regularizers have been explored for promoting sparsity as well as related low-rank promoting regularizers \cite{ravishankar2019image}. \textcolor{black}{Other methods depend on traditional model-based reconstruction such as combining Plug-and-Play (PnP) \cite{venkatakrishnan2013plug} and Block-Matching and 3D Filtering (BM3D) \cite{4271520}.} \textcolor{black}{These methods typically require task-specific designs and handcrafted priors that may be unsuitable for certain settings. However, handcrafted priors are often overly simplistic and may not be able to capture the rich structures of natural images \cite{ravishankar2019image}. Therefore, with the advancement of machine learning tools, different approaches have been explored as we describe next.}



In recent years, \textcolor{black}{the learning of sparsity-promoting models has shown promise for IIPs~\cite{ravishankar2019image}. Moreover,} numerous Deep Neural Network (DNN) techniques have been developed to address IIPs as illustrated in \textcolor{black}{most of} Fig.~\ref{fig: timeline}. These techniques include supervised models \textcolor{black}{(such as end-to-end CNNs \cite{8103129} and deep unrolling \cite{8434321})}, generative models \textcolor{black}{(such as Diffusion posterior sampling, DPS~\cite{chung2022diffusion} and Decomposed Diffusion Sampling, DDS \cite{chung2024decomposed})}, and \textcolor{black}{self-supervised methods (such as Noise2Noise \cite{lehtinen2018noise2noise})}. Despite their effectiveness, such models generally require extensive amounts of data for training, which limits their applicability in training-data-limited tasks, including but not limited to medical applications (e.g., magnetic resonance imaging (MRI) and computed tomography (CT)). \textit{This challenge highlights the need for methods that can reduce the reliance on large, fully-sampled (or labeled) datasets and/or pre-trained models.}

Under the limited training data setting, 
early popular approaches included those 
relying on adapted models such as patch-based dictionaries and sparsifying transforms, often estimated from only measurements, and using them to reconstruct underlying images~\cite{ravishankar2015efficientTransform,Sai_Dict}. More recently, these methods have been extended to estimate neural networks without training data. One notable approach is deep image prior (DIP)~\cite{ulyanov2018deep}, a method that operates without pre-trained models, instead leveraging the parameters of a deep convolutional neural network architecture (e.g., U-Net \cite{ronneberger2015u}) within a training-data-less setting \cite{alkhouri2022differentiable}, optimizing instead costs like~\eqref{eqn: main}. \textcolor{black}{More specifically, DIP re-parameterizes the optimization problem in \eqref{eqn: main} using a deep untrained network. 
In DIP~\cite{ulyanov2018deep}, it was \textit{empirically} demonstrated} that the architecture of a generator network alone is capable of capturing a significant amount of low-level image statistics even before any learning takes place. 

Although DIP has shown significant potential for solving various inverse imaging problems, it (and most of its variants) faces challenges with noise \textcolor{black}{over-fitting}~\cite{alkhouriNeuIPS24} as DIP first learns the natural image component of the corrupted image but gradually overfits the noise due to its highly over-parameterized nature -- a phenomenon known as spectral bias~\cite{chakrabarty2019spectral}. Consequently, the optimal number of optimization steps (before \textcolor{black}{over-fitting} worsens) varies not only by task but also by image for the same task and distribution. As optimization progresses, the network's output tends to fit the noise present in the measurements and may also fit to undesired images within the null space of the imaging measurement operator. An example of these phenomena for the tasks of natural image denoising and undersampled MR image reconstruction are shown in Fig. \ref{fig:overfitting}.
\begin{figure}[t]
\centering
\includegraphics[width=1\linewidth]{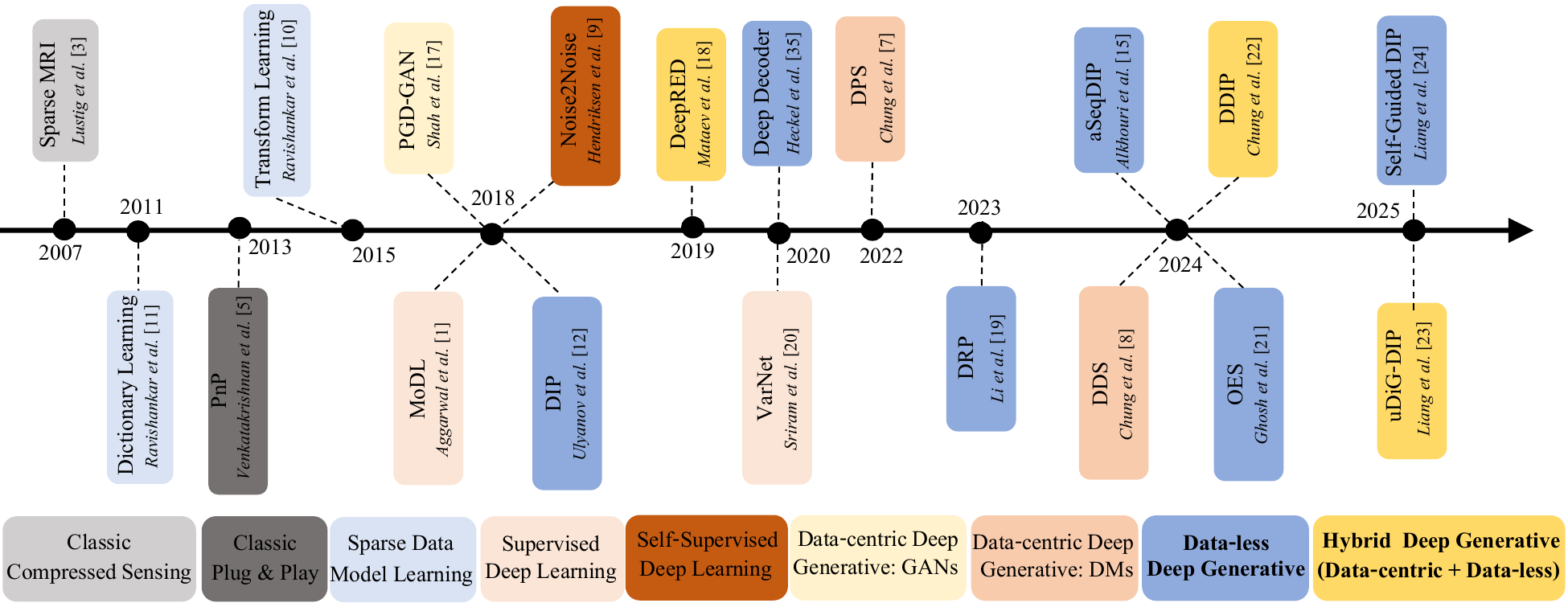}
\vspace{-0.8cm}
\caption{\textbf{Timeline of the development of approaches for inverse imaging problems \textcolor{black}{starting from classical approaches to learning and neural networks-based methods}}. The bottom row categorizes each method in the top two rows under a broader approach. The bolded approaches in the bottom right represent the categories covered in this tutorial. This figure is best viewed in color. From left to right, the acronyms are: Sparse MRI \cite{lustig2007sparse}, \textcolor{black}{Plug and Play (PnP) \cite{venkatakrishnan2013plug}}, Projected Gradient Descent Generative Adversarial Networks (PGD-GAN) \cite{PGD-GAN}, Model-based Deep Learning (MoDL) \cite{8434321}, Deep Image Prior (DIP) \cite{ulyanov2018deep}, \textcolor{black}{Noise to Noise (Noise2Noise) \cite{lehtinen2018noise2noise}}, Deep Regularization by Denoising (DeepRED) \cite{Mataev_2019_ICCV}, \textcolor{black}{Deep Random Projector (DRP) \cite{Li_2023_CVPR},} Variational Network (VarNet) \cite{sriram2020end}, Diffusion Posterior Sampling (DPS) \cite{chung2022diffusion}, \textcolor{black}{\textcolor{black}{Decomposed Diffusion Sampling (DDS) \cite{chung2024decomposed}}, Optimal Eye Surgeon (OES) \cite{ghosh2024optimal}, Autoencoding Sequential DIP (aSeqDIP) \cite{alkhouriNeuIPS24}, \textcolor{black}{Diffusion Deep Image Prior (DDIP) \cite{ddip}}, and Sequential Diffusion Guided DIP (uDiG-DIP) \cite{LianguDigDIP}. Hybrid deep generative methods, which integrate DIP with other pre-trained models, will be further discussed in Section~\ref{sec: combining DIP with other models}.}}
\label{fig: timeline}
\vspace{-0.4cm}
\end{figure}
\begin{figure}[t]
\centering
\includegraphics[width=1\linewidth]{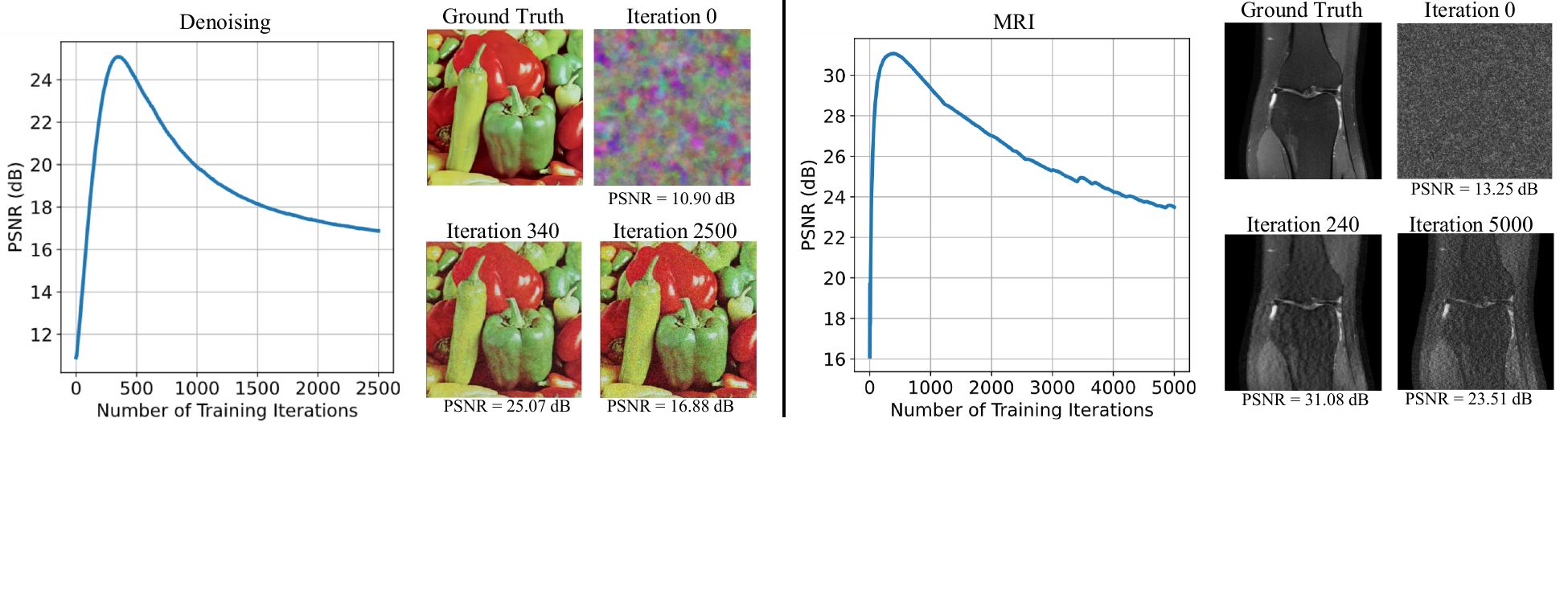}
\vspace{-2.9cm}
\caption{\textbf{Demonstration of the \textcolor{black}{over-fitting} phenomenon in DIP}. Left: Gaussian denoising with $\sigma=50$ using the ``Peppers" test image. Right: $4\times$ accelerated MRI reconstruction using an image from the \texttt{fastMRI} knee dataset. In each task, the three images show the initial network output (iteration 0), the output at the peak of the PSNR curve (iteration 340 for denoising and iteration 240 for MRI), and the output at the end of the optimization (iteration 2500 for denoising and iteration 5000 for MRI). For denoising, the network overfits to the noise in the target image; in MRI, the network eventually outputs spurious frequency content in the null space of the forward operator.}
\label{fig:overfitting}
\vspace{-0.4cm}
\end{figure}
Towards addressing the noise \textcolor{black}{over-fitting} issue, several approaches have been proposed which conceptually can be categorized into three categories: regularization (such as \cite{liang2024analysis}), network re-parameterization (such as \cite{ghosh2024optimal,you2020robust}), and early stopping (such as \cite{wang2021early}). DIP has been applied to various IIPs, including MRI \cite{alkhouriNeuIPS24}, CT \cite{liang2024analysis}, and several image restoration tasks \cite{wang2021early}, achieving highly competitive (and sometimes leading) qualitative results. For example, the work in \cite{alkhouriNeuIPS24} (resp. \cite{wang2021early}) demonstrated that DIP variants on MRI can outperform data-centric generative (resp. supervised) methods in terms of the reconstruction quality --\textit{all without requiring any training data}. The concept of solving inverse problems in a training-data-free regime using network-based priors has been extended to solving dynamic (e.g., video) inverse problems \cite{NEURIPS2020_0c0a7566}, where the authors showed that the architecture of the network can achieve an improved blind temporal consistency. Towards understanding the optimization dynamics in DIP, multiple studies have considered the Neural Tangent Kernel (NTK) \cite{liang2024analysis,tachella2021neural}, which is a tool used to analyze the training dynamics of neural networks in the infinite width limit. In the NTK regime, updates take place mostly in the top eigenspaces of the kernel leading to smooth approximations of the image. Deep networks have an ``implicit bias'' for reconstructing low frequency parts of the image before they overfit to the noise (spectral bias~\cite{chakrabarty2019spectral}).

\noindent\textbf{Contributions}: In this \textcolor{black}{tutorial article}, we first present the DIP framework, and \textcolor{black}{discuss} its fundamental issue with respect to noise \textcolor{black}{over-fitting}. Subsequently, we describe the theoretical results/tools used to explain/justify deep image prior. Second, we review key recent works that address the noise \textcolor{black}{over-fitting} issues and describe additional method-specific theoretical results for natural and medical imaging problems (e.g., fitting the null space in the forward operator in MRI \cite{liang2024analysis} and the need for double priors for phase retrieval \cite{zhuang2023practical}). Third, we review recent methods that combine DIP with other pre-trained models. We also present empirical comparisons and insights. Finally, we describe gaps in theory for 
recent directions (open questions and future directions). Our goal is to create awareness and accelerate research in the topic of DIP among researchers in computational imaging applications, signal processing, optimization, and machine learning. The tutorial is also used to educate and encourage more scholars from multi-disciplinary fields to exploit the Deep Image Prior and related paradigms from the image processing community. 



\subsection{\textcolor{black}{Related Works}}

Here, we first discuss two recent papers that focused on topics close to this paper. \textcolor{black}{Then, we describe the similarities and differences between DIP and Implicit Neural Representation (INR).}

In \cite{lu2022priors}, the authors surveyed previous studies that use DL-based priors for image restoration and enhancement, including DIP. In addition to DIP, the authors covered data-centric DNN methods where the prior is either a pre-trained GAN or a pre-trained supervised model. Our paper differs in that it focuses more on existing theoretical results and open questions, and second, in the fact that we cover more categories (for noise \textcolor{black}{over-fitting} prevention) than \cite{lu2022priors} (where the authors only covered two types of methods). The closest work to our paper is \cite{9878048}, where the authors presented the first paper that comprehensively covers applications of DIP (or Untrained Neural Network Priors) for inverse imaging problems (medical and natural images) up to late 2022. There are two main differences between our work and \cite{9878048}. First and perhaps more importantly, in our paper, we focus on key theoretical results and analysis of DIP and categorize studies based on addressing the noise over-fitting issue (instead of applications) along with their method-specific theoretical results. Additionally, the open questions and future directions focus on identifying theoretical questions and gaps. Second, we focus on 
several more recent works compared to~\cite{9878048}. 

\noindent\textbf{\textcolor{black}{Relation to Implicit Neural Representations Networks}}: The work in \cite{INR_Liyie} introduced Implicit Neural Representation (INR) with Prior embedding (NeRP) for IIPs. While both NeRP and DIP optimize the parameters of a NN for image reconstruction, they differ in three key aspects. First, NeRP employs a neural network without convolutions (i.e., multilayer perceptron (MLP)), whereas DIP typically relies on CNNs. Second, the mapping functions are very different as NeRP maps from spatial coordinates to function/image values. Third, \textcolor{black}{many NeRP-based methods} require a prior image, as training takes place in two stages: (\textit{i}) the prior image is used to train an MLP; and (\textit{ii}) the typical data-fitting loss is used on another MLP to find the reconstructed image for which the initialization of the weights of the second MLP is based on the pre-trained MLP on the prior image (from the first stage). \textcolor{black}{We note that there have been studies that are not prior-informed INR such as the method in \cite{9606601}. However, these methods still require the specific spatial input embeddings and the use of non-convolutional architectures. } 

\section{Deep Image Prior: Framework, Challenges, \& Theory}

\subsection{\textcolor{black}{Basics of} Deep Image Prior}


Let $f : \mathbb{R}^l \rightarrow \mathbb{R}^n$ be a convolutional NN with parameters $\boldsymbol{\theta}$ (typically selected as a U-Net with residual connections \cite{ronneberger2015u}). The DIP image reconstruction framework re-parameterizes the optimization problem in \eqref{eqn: main} with $\mathbf{x} = f_\theta(\mathbf{z})$, where $\mathbf{z} \in \mathbb{R}^n$ is typically randomly chosen. Due to this over-parameterization, the structure of the CNN provides an implicit prior \cite{ulyanov2018deep} (i.e., the second term in \eqref{eqn: main} is neglected). Therefore, the \textit{training-data-free} DIP image reconstruction is obtained through the minimization of the following objective:
\begin{equation}
\label{eqn: standard DIP}
    \hat{\boldsymbol{\theta}} = \argmin_{\boldsymbol{\theta}}  ~ \frac{1}{2} \| \mathbf{A} f_{\boldsymbol{\theta}}(\mathbf{z}) - \mathbf{y} \|_2^2 \:,\;\;\;\; \hat{\mathbf{x}}=  f_{\hat{\boldsymbol{\theta}}}(\mathbf{z})\:,
\end{equation}
where $\hat{\mathbf{x}}$ is the reconstructed image. We note that we use a linear forward operator for ease of notation. We refer to the optimization problem \eqref{eqn: standard DIP} as ``vanilla DIP" throughout the remainder of the paper.


\subsection{\textcolor{black}{DIP Challenges}}

In DIP, selecting the number of iterations of an algorithm to optimize the objective in~\eqref{eqn: standard DIP} poses a challenge as the network would eventually fit the noise present in $\mathbf{y}$ or could fit to undesired images based on the null space of $\mathbf{A}$. An empirical demonstration of this phenomenon for image denoising and MRI reconstruction is shown in Fig.~\ref{fig:overfitting}. The optimal stopping iteration can vary not only from task to task, but also between images within the same task, dataset, and semantics. To further demonstrate this, for the denoising task, we pick 4 ``cat'' images from the ImageNet dataset, and run the optimization in \eqref{eqn: standard DIP} with the exact same initialization (i.e., for $\mathbf{z}$ and $\boldsymbol{\theta}$) and optimizer (where we use Adaptive Moment Estimation (Adam)) with learning rate $10^{-4}$. Fig.~\ref{fig: psnr_curve_imagenet} presents the Peak Signal to Noise Ratio (PSNR) curves over optimization iterations and highlights the maximum PSNR for each image (represented by the dotted vertical lines). As observed, even with the same architecture, initialization  for theta, and same image semantics, the optimal number of steps to optimize is not the same.


Another challenge of DIP is its computational cost at inference, as a separate optimization problem must be solved for each measurement vector $\mathbf{y}$. However, this issue has received less attention, as, to our knowledge, only one method—Deep Random Projector \cite{Li_2023_CVPR}—has explicitly addressed the slow inference problem, which we discuss in the next section.

\textcolor{black}{We note that DIP is not the only DL-based method that is slow at inference as several data-centric diffusion models (DMs) based IIP solvers (such as DPS \cite{chung2022diffusion} and DDS \cite{chung2024decomposed}) also require long inference times due to their iterative sampling procedure. An example of how a variant of DIP can be faster than some DM-based methods is given in Table~2 in \cite{alkhouriNeuIPS24}.} \textcolor{black}{When compared to non-diffusion data-centric methods, such as the ones based on supervised learning \cite{8434321,8103129,sriram2020end}, the run-time of DIP is expected to be much higher than the inference run-time of supervised deep network methods. The reason is that the supervised methods require one or very few forward passed through the trained network. However, DIP does not require any training data or pre-trained models and optimizes at inference time.}


%
\begin{figure}[t]
\centering
\includegraphics[width=1\linewidth]{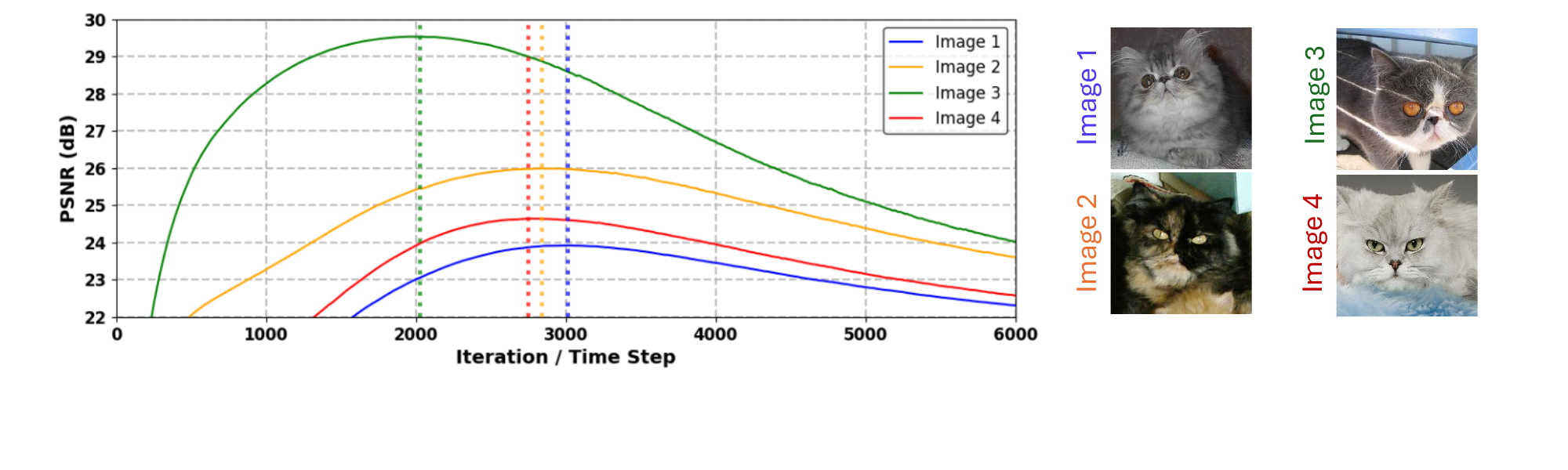}
\vspace{-2.0 cm}
\caption{PSNR curves for four ``cat'' images from the ImageNet dataset over iterations of the optimization in \eqref{eqn: standard DIP} for the task of denoising. The measurements $\mathbf{y}$ were obtained by adding a perturbation vector drawn from a Gaussian distribution with zero mean and $25/255$ standard deviation. Vertical colored lines mark the argmax, highlighting that even within the same task, dataset, and semantics, the optimal number of iterations to reach peak PSNR varies.}
\label{fig: psnr_curve_imagenet}
\vspace{-0.7cm}
\end{figure}
%

\subsection{Theoretical Approaches to Understanding DIP}

The success of deep image prior in the field of image reconstruction has motivated theoretical studies on how over-parameterization and architecture bias affect the implicit bias. We broadly classify the existing studies into two categories: those that use the \textit{neural tangent kernel} (NTK) to simplify the training dynamics of DIP, and those that interpret DIP through the lens of \textcolor{black}{over-parameterized} matrix factorization. \textcolor{black}{In both these  settings, it has been shown that Gradient descent biases the solution towards favorable solutions that are biased to low-rank solutions, an inherent property of natural images. The NTK analysis, often referred to as ``lazy learning'', essentially describes a kernel regression with a fixed kernel evaluated at initialization. A notable disadvantage of this approach is its inability to promote feature learning due to its fixed kernel. Hence, we also discuss the implicit bias of DIP from a two-layer matrix factorization framework, which promotes feature learning by showing how the network actively discovers structured representations, corresponding to low-rank learning.}

We summarize the assumptions, main results, and limitations of these two methods in Table \ref{tab:ntk_vs_matrix_factorization}.

\begin{table}[h]
    \centering
    \begin{tabular}{|l|m{6cm}|m{6cm}|}
        \hline
        \textbf{Method} & \textbf{NTK Analysis} & \textbf{Matrix Factorization Analysis} \\
        \hline
        \textbf{Network Structure} & Linearized around initialization $f_\btheta(\z) = f_{\btheta_0}(\z) +  \mathbf{J} (\boldsymbol{\theta} -\boldsymbol{\theta}_{0})$ & Low-rank matrix factorization (\(\mathbf{X} = \mathbf{U}\mathbf{U}^\top\)) \\
        \hline
        \textbf{Assumptions on Network} & Very large width, $\mathbf{J}$ is full row rank and non-trivial null space & Hidden layer dimension much larger than true rank of the signal ($r>>n$). \\
        \hline
        \textbf{Implicit Bias} & Bias towards smooth, low-frequency reconstructions due to spectral decay. In \cite{heckel2020compressive}: With very high probability, $\|\hat{\mathbf{x}} - \mathbf{x}^*\|_2^2 \leq C \left( \sum_{i=1}^{n} \frac{1}{\sigma_i^2} \langle \mathbf{w}_i, \mathbf{x}^* \rangle^2 \right) \sum_{i > 2m/3} \sigma_i^2.$ In \cite{liang2024analysis}: three cases, bias depends on the relationship between $\A$ and $\mathbf{JJ}^\top$. & Bias towards low-rank solutions via implicit nuclear norm minimization, $   \min_{\mathbf{X} \succeq \mathbf{0}} \|\mathbf{X}\|_* \quad \text{s.t.} \quad \mathbf{A}(\mathbf{X}) = y.$\\
        \hline
        \textbf{Assumptions on forward operator} & In \cite{heckel2020compressive}: obeys restricted isometry property. In \cite{liang2024analysis}: full row rank. & Obeys restricted isometry property. The measurement operators are symmetric and commutative. \\
        \hline
        \textbf{Hyperparameter assumption} & Random Gaussian initialization and stable finite step-size $\eta < \frac{2}{\| \mathbf{J}\|_{2}}$ & Infinitesimally small random initialization and gradient flow ($\eta \rightarrow 0$).  \\
        \hline
        \textbf{Limitations} & Requires the NTK assumption, which is only valid for very wide networks. It is difficult to determine structure and properties of the NTK analytically for networks beyond two layers. & Requires two layer linear network assumption. Analysis partially breaks down with non-linear activation functions. Does not count in the effect of large width. \\
        \hline
    \end{tabular}
    \vspace{0.15cm}
    \caption{Comparison of NTK Analysis and Matrix Factorization Analysis in Deep Image Prior.}
    \vspace{-0.2in}
    \label{tab:ntk_vs_matrix_factorization}
\end{table}





\subsubsection{\textcolor{black}{Neural Tangent Kernel Analysis}}

The first line of work, represented by the theory developed in \cite{liang2024analysis,tachella2021neural}, analyzes DIP by linearizing the network $f$ around its initialization.
In particular, they assume that the network's output for a particular set of parameters $\boldsymbol{\theta}$ is well approximated by a first order Taylor expansion around the initialization $\boldsymbol{\theta}_0$:
\begin{equation}\label{eqn: NTK first eqn}
f_\btheta(\z) = f_{\btheta_0}(\z) +  \mathbf{J} (\boldsymbol{\theta} -\boldsymbol{\theta}_{0}),    
\end{equation}
where $\mathbf{J}$ is the Jacobian of $f$ with respect to $\boldsymbol{\theta}$ evaluated at $\boldsymbol{\theta}_0$, i.e. $\mathbf{J} := \nabla_\btheta f_\btheta(\z)\big\vert_{\btheta=\btheta_0}$. Under technical assumptions about the initialization of the network parameters, in the limit of infinite network width (where for CNNs ``width" corresponds to number of channels), this linearization holds exactly \cite{tachella2021neural}.

When optimizing the DIP objective in \eqref{eqn: standard DIP} with gradient descent, this linearization implies that the change in the network output at iteration $t$ is given by:
\begin{equation}
    \label{eq:lin_at_iter_t}
    f_{\btheta_{t+1}}(\z) = f_{\btheta_t}(\z) + \mathbf{J} (\btheta_{t+1} - \btheta_t),
\end{equation}
and, when using gradient descent with a step size of $\eta$ to minimize the objective in \eqref{eqn: standard DIP}, the update to the parameters is given by: 
%
\begin{align}
    \label{eq:param_diff}
    \btheta_{t+1} - \btheta_t &= \frac{-\eta}{2}\, [\nabla_\btheta ||\A f_\btheta(\z) - \y||_2^2]\big|_{\btheta=\btheta_t} = \eta \, (\nabla_\btheta f_\btheta(\z)\big\vert_{\btheta=\btheta_t})^\top(\A^\top\y - \A^\top\A f_{\btheta_t}(\z)) \nonumber \\
    &\overset{(*)}{=} \eta \,\mathbf{J}^\top(\A^\top\y - \A^\top\A f_{\btheta_t}(\z)),
\end{align}
where $(*)$ follows from the assumption that the network Jacobian does not change from its initialization. Finally, substituting the result of \eqref{eq:param_diff} into equation \eqref{eq:lin_at_iter_t} gives the following recursion for the network output at every iteration:
\begin{equation}
    \label{eq:jac-update}
    f_{\btheta_{t+1}}(\z) = f_{\btheta_t}(\z) + \eta \mathbf{J}\mathbf{J}^\top(\A^\top\y - \A^\top\A f_{\btheta_t}(\z)).
\end{equation}
The matrix $\mathbf{J}\mathbf{J}^\top$ is known as the \textit{neural tangent kernel}, and we will denote it as $\mathbf{K}$ in subsequent sections.

\textcolor{black}{We now present theoretical results based on the NTK perspective from two recent works. The first is \cite{heckel2020compressive}, which demonstrates how spectral bias emerges in image reconstruction with untrained networks for a simple two-layer architecture. In this analysis, the known spectral properties of the matrix $\mathbf{J}$ play a key role in how the optimization in \eqref{eqn: standard DIP} accurately recovers smooth signals from few measurements. We then present complementary results from \cite{liang2024analysis}, which provides a generic analysis of the iterates in equation \eqref{eq:jac-update}. This approach applies to more general network architectures, and the results reveal the importance of the relationship between the NTK and the forward operator $\A$ in enabling successful signal recovery.}

Heckel and Soltanolkotabi \cite{heckel2020compressive} analyze the NTK updates for a two layer version of the \textit{deep decoder} \cite{heckel2019deep}, a class of untrained convolutional neural networks, given by $f(\boldsymbol{\theta}) = \text{ReLU}(\mathbf{U} \boldsymbol{\theta})\mathbf{v} $
where \( \mathbf{U} \in \mathbb{R}^{n \times n} \) is a fixed convolutional operator, \( \boldsymbol{\theta} \in \mathbb{R}^{n \times k} \) is a parameter matrix, and \( \mathbf{v} \in \mathbb{R}^{k} \) is a fixed output weight vector. This model is over-parameterized, meaning \( k \gg n \), and is used in inverse problems such as compressive sensing. This two layer deep decoder does not have an explicit input, but the parameter matrix $\btheta$ can be thought of as representing the output of the first layer of a CNN with a fixed input, i.e., the fixed input $\z$ has been absorbed into the network weights $\btheta$ and is suppressed in the notation.

A key observation in \cite{heckel2020compressive} is that the Jacobian \( \mathbf{J} \) associated with this convolutional generator exhibits a highly structured spectral decomposition. Specifically, the left singular vectors of \( \mathbf{J} \) are well approximated by trigonometric basis functions (ordered from low to high frequencies), and the singular values \( \sigma_i \) decay \textit{geometrically}, i.e., $\sigma_i^2 \approx \gamma^i$ for some  $0 < \gamma < 1$. \textcolor{black}{Intuitively, we can then expect} that \textcolor{black}{when the filtering in \eqref{eq:jac-update} is performed, \textit{the}} \textit{high-frequency components \textcolor{black}{of the signal} are suppressed}, leading to an implicit spectral bias.

\textcolor{black}{We now explain the emergence of this spectral bias in more detail, and demonstrate} how this spectral bias impacts signal recovery. \textcolor{black}{To exploit our knowledge about the spectral decomposition of $\mathbf{J}$, we rewrite it} using its singular value decomposition as $\mathbf{J} = \mathbf{W} \mathbf{\Sigma} \mathbf{V}^\top$. We then consider a signal $\mathbf{x}^*$, which can then be decomposed as $\mathbf{x}^*  = \sum_{i=1}^{n} \langle \mathbf{x}^* , \mathbf{w}_i \rangle \mathbf{w}_i$, where $\mathbf{w}_i$ are the left singular vectors of $\mathbf{J}$. For the two-layer deep decoder $f(\boldsymbol{\theta})$, the left-singular vectors of $\mathbf{J}$ are well approximated by the trigonometric basis. \textcolor{black}{Hence, if $\mathbf{x}^*$ is reasonably smooth, we expect that it can be represented well in this basis with a relatively small number of components.}

\textcolor{black}{We now connect this insight to signal recovery by} considering the update \eqref{eq:jac-update}. \textcolor{black}{We also crucially use} the fact that small step-size on this least squares regression loss (with the initialization $\mathbf{c}_0 =\mathbf{0}$) leads to the minimum norm solution as follows: $\hat{\mathbf{c}} = \arg\min_{\mathbf{c}}  \|\mathbf{c} \|_{2}^2 \quad \text{subject to}  \quad \mathbf{A}\mathbf{J}\mathbf{c} = \mathbf{y},$
which has a closed form solution $ \hat{\mathbf{c}} = \mathbf{P}_{\mathbf{J}^\top \mathbf{A}^{T} }\mathbf{c}^{*}$, where $\mathbf{P}_{\mathbf{J}^\top \mathbf{A}^{T} }$ denotes the orthogonal projection operator onto the range of $(\mathbf{A}\mathbf{J})^{T}$ and $\mathbf{c}^{*}$ denotes any solution which generates the ground truth as $\mathbf{x}^{*} = \mathbf{J}\mathbf{c}^{*}$. We assume that $\mathbf{J}$ is full rank so a $\mathbf{c}^*$ exists. Then the signal estimation error is 
$\hat{\mathbf{x}} - \mathbf{x}^{*} = \mathbf{J}(\mathbf{\hat{c} - \mathbf{c}^{*}) = \mathbf{J}(\mathbf{P}_{\mathbf{J}^{T}\mathbf{A}^{T}} - \mathbf{I{}})\mathbf{c}^{*}}\:.$ 
We then expect the error in estimating $\mathbf{x}^{*}$ from the compressive measurements $\mathbf{y}=\mathbf{A}\mathbf{x}^{*}$ to be small under two conditions:
\textcolor{black}{(\textit{i})} $\mathbf{x}^{*}$ approximately lies in the span of the leading singular vectors
of $\mathbf{J}$ and \textcolor{black}{(\textit{ii})} the singular values of the matrix $\mathbf{J}$ decay sufficiently fast. \textcolor{black}{Condition (\textit{i}) is again primarily related to the smoothness of $\mathbf{x}^*$, since satisfying this condition means that $\mathbf{x}^*$ can be accurately represented by a small number of low-frequency basis functions. Assuming that this condition is satisfied,} it is reasonable to expect that there may be a particular $\mathbf{c}^*$ with small norm, because a relatively small number of coefficients are needed to accurately represent $\x^*$ in the column space of $\mathbf{W}$. \textcolor{black}{Condition (\textit{ii}) is important because it is related to the alignment between the recovered $\hat{\mathbf{c}}$ and this particular $\mathbf{c}^*$}. \textcolor{black}{In particular, if this condition is satisfied, we} expect that $\hat{\mathbf{c}}$ and $\mathbf{c}^*$ will be closely aligned. We expect this because (assuming, e.g., that $\A$ obeys the restricted isometry property) significant differences between $\hat{\mathbf{c}}$ and $\mathbf{c}^*$ will generically live in the subspace corresponding to the trailing singular values of $\mathbf{J}$ for the constraint $\A\mathbf{Jc}=\mathbf{y}$ to remain satisfied. However, such differences will be highly penalized since $\hat{\mathbf{c}}$ is the minimum-norm solution, so we can expect $\hat{\mathbf{c}} \approx \mathbf{c}^*$.

\textcolor{black}{These two conditions for accurate signal recovery may explain the strong empirical performance of the deep decoder, because} for the two layer deep decoder, the dominant left singular vectors of $\mathbf{J}$ are low-frequency trigonometric basis functions and $\mathbf{J}$'s singular values tend to decay geometrically. Since the frequency spectrum of natural images also exhibits a fast decay, we expect that both conditions will be satisfied, leading to effective signal recovery. The main theorem of \cite{heckel2020compressive} makes these notions rigorous in the case of compressed sensing with a Gaussian measurement matrix:
 
\begin{theorem}[\textcolor{black}{Theorem 1 from \cite{heckel2020compressive}}]
\label{spectral-gradient}
Let \( \mathbf{A} \in \mathbb{R}^{m \times n} \) be a random Gaussian measurement matrix with \( m \geq 12 \), and let \( \mathbf{w}_1, \dots, \mathbf{w}_n \) be the left singular vectors of the neural network Jacobian \( \mathbf{J} := \nabla_\btheta f(\btheta)\big\vert_{\btheta=\btheta_0} \), where $\btheta_0$ are the initial parameters. Let the corresponding singular values of $\mathbf{J}$ be \( \sigma_1 \geq \dots \geq \sigma_n \). Then, for any \( \mathbf{x}^* \in \mathbb{R}^n \), with probability at least \( 1 - 3e^{-1/2m} \), the reconstruction error using a gradient-descent-trained convolutional generator satisfies:
\[
\|\hat{\mathbf{x}} - \mathbf{x}^*\|_2^2 \leq C \left( \sum_{i=1}^{n} \frac{1}{\sigma_i^2} \langle \mathbf{w}_i, \mathbf{x}^* \rangle^2 \right) \sum_{i > 2m/3} \sigma_i^2.
\]
where \( C \) is a universal constant. 
\end{theorem}


\textcolor{black}{As previously explained intuitively, the} main message of Theorem~\ref{spectral-gradient} is that the error in the recovered signal is controlled by two factors: the ``smoothness" of the signal $\x^*$ (how easily it is represented by the leading left singular vectors of $\mathbf{J}$) and the decay of the singular values $\sigma_i$. The sum $\sum_{i=1}^{n} \frac{1}{\sigma_i^2} \langle \mathbf{w}_i, \mathbf{x}^* \rangle^2$ will be small provided that $\langle \mathbf{w}_i, \mathbf{x}^* \rangle$ is small for larger $i$, i.e. $\mathbf{x}^*$ predominantly lies in the space spanned by the leading left singular vectors of $\mathbf{J}$, which are low-frequency trigonometric basis functions in the case of the deep decoder, implying that $\mathbf{x}^*$ is smooth. The second term $\sum_{i > 2m/3} \sigma_i^2$ will be small if $\mathbf{J}$ exhibits fast singular value decay (for example, geometric decay in the case of the deep decoder).





While the results of \cite{heckel2020compressive} reveal how the implicit bias of convolutional generators can enable signal recovery, this analysis is limited to compressed sensing with a Gaussian measurement matrix. Additional works \cite{tachella2021neural,liang2024analysis} directly analyze the recursive updates in equation \eqref{eq:jac-update} to obtain signal recovery guarantees. These approaches enable the use of more general forward operators $\mathbf{A}$. When analyzing equation \eqref{eq:jac-update}, it is natural to believe that the relationship between the NTK $\mathbf{K} = \mathbf{JJ}^\top$ and the forward operator $\mathbf{A}$ will play a crucial role in the dynamics of signal recovery with DIP. \textcolor{black}{Indeed, in} \cite{liang2024analysis}, it is shown that there are 3 regimes for signal recovery under the NTK assumption, depending on this relationship. We now state the main theorem from \cite{liang2024analysis}, which describes the three cases. \textcolor{black}{We then provide intuition about when these conditions may be approximately satisfied in realistic scenarios.} A particularly intriguing feature of this analysis is that it is able to provide a condition for exact recovery of the underlying signal.

\begin{theorem}[\textcolor{black}{Theorem~1 from \cite{liang2024analysis}}]
    \label{thm:noiseless_recovery}
    Let $\mathbf{A} \in \mathbb{R}^{m \times n}$ be of full row rank. Suppose that $f_{\boldsymbol{\theta}_0}(\z) = \boldsymbol{0}$, and let $f_{\boldsymbol{\theta}_\infty}(\z)$ be the reconstruction as the number of gradient updates approaches infinity. Let $\mathbf{x} \in \mathbb{R}^n$ be the true signal and the measurements are assumed noise-free so that $\mathbf{y} = \mathbf{Ax}$. If the step size $\eta < \frac{2}{||\mathbf{B}||}$, where $\mathbf{B} := \mathbf{K}^{1/2}\mathbf{A}^\top\mathbf{A}\mathbf{K}^{1/2}$, where the $\mathbf{K}$ is the NTK, i.e. $\mathbf{K} := \left( \nabla_\btheta f_\btheta(\z)\big\vert_{\btheta=\btheta_0} \right)\left( \nabla_\btheta f_\btheta(\z)\big\vert_{\btheta=\btheta_0} \right)^\top$
    , then:
    \begin{enumerate}
        \item If the NTK $\mathbf{K}$ is non-singular, then the difference $f_{\boldsymbol{\theta}_\infty}(\z) - \mathbf{x} \in N(\mathbf{A})$, where $N(\mathbf{A})$ denotes the null space of $\mathbf{A}$. Moreover, provided the projection $P_{N(\mathbf{A})}\mathbf{x} \neq \boldsymbol{0}$, the error $f_{\boldsymbol{\theta}_\infty}(\z) - \mathbf{x} \neq \boldsymbol{0}$.
        \item If $\mathbf{K}$ is singular and $P_{N(\mathbf{A}) \cap R(\mathbf{K})} \mathbf{x} = \boldsymbol{0}$, then the error $f_{\boldsymbol{\theta}_\infty}(\z) - \mathbf{x}$ will depend only on $P_{N(\mathbf{K})}\mathbf{x}$, in particular $f_{\boldsymbol{\theta}_\infty}(\z) - \mathbf{x} = -P_{N(\mathbf{K})}\mathbf{x} + \mathbf{K}(\mathbf{A}\mathbf{K}^{1/2})^\dagger\mathbf{A}P_{N(\mathbf{K})}\mathbf{x}$, where $ R(\mathbf{K})$ denotes the range space of $\mathbf{K}$. 
        \item If $\mathbf{K}$ is singular, $P_{N(\mathbf{A}) \cap R(\mathbf{K})} \mathbf{x} = \boldsymbol{0}$, and $\mathbf{x}\in R(\mathbf{K})$, then the reconstruction is exact, with $f_{\boldsymbol{\theta}_\infty}(\z) = \mathbf{x}$.
    \end{enumerate}
\end{theorem}

\textcolor{black}{The result in the first case tells us that if the NTK is non-singular, the error in the limiting reconstruction lies entirely in the null space of the forward operator. However, this is also true for many simple reconstruction methods, such as the pseudoinverse reconstruction $\A^\dagger\y$. The more interesting observation is that it also tells us that in this case the signal $\x$ can \textit{never} be perfectly recovered if any part of $\x$ lies in the null space of $\A$. In the second case, we can again intuit that $\x$ having little content in $N(\A)$ may be important for accurate signal recovery (specifically the subspace $N(\A) \cap R(\mathbf{K})$). Moreover, the null space of the NTK now also plays an important role, and we would expect a small error in the recovered signal if $P_{N(\mathbf{K})}\x$ is small. The third case provides a condition for exact signal recovery. It is effectively a corollary of the second case, since $R(\mathbf{K})$ is the orthogonal complement of $N(\mathbf{K})$ because $\mathbf{K}$ is symmetric.}

\textcolor{black}{This theorem tells us that exact signal recovery requires two important conditions. First, the NTK must be able to accurately represent $\x$ (i.e. $\x \in R(\mathbf{K})$). Moreover, the NTK and the forward operator $\A$ must be sufficiently \textit{incoherent} with each other, i.e. the subspaces $N(\A)$ and $R(\mathbf{K})$ are (mis)aligned such that $P_{N(\mathbf{A}) \cap R(\mathbf{K})} \mathbf{x} = \boldsymbol{0}$. A simple example of a case where these conditions could hold is as follows (adapted from \cite{liang2024analysis}). Suppose that $\x$ consists of a small number of non-bandlimited wavelet components, and that the NTK can represent this signal ($\x \in R(\mathbf{K})$), and suppose that $\A$ samples a range of low-frequency Fourier modes. If the wavelets that make up $\x$ cannot be linearly combined to form a bandlimited signal (which would be in $N(\A))$, then one would have $P_{N(\mathbf{A}) \cap R(\mathbf{K})} \mathbf{x} = \boldsymbol{0}$, enabling exact recovery.}

\textcolor{black}{We further note that the conditions of incoherence and the ability to represent $\x$ in $R(\mathbf{K})$ are very similar to the assumptions of Theorem \ref{spectral-gradient}, which relies on the restricted isometry property (guaranteeing some level of incoherence between $\A$ and $\mathbf{J}$ or $\mathbf{K}$) and the ability to compactly represent $\x$ in the leading singular vectors of the network Jacobian $\mathbf{J}$, which has the same range space as $\mathbf{K}$. The fact that these assumptions appear in both theoretical analyses further underscores their importance for understanding the mechanisms present in DIP.}

\textcolor{black}{However, while both Theorems \ref{spectral-gradient} and \ref{thm:noiseless_recovery} provide important insights into how DIP enables signal recovery, neither of these theorems addresses the more realistic case where the measurements $\y$ are corrupted by noise.}
The \textcolor{black}{over-fitting} of DIP in the presence of noise can be interpreted as the bias-variance tradeoff present in classical image filtering algorithms. Indeed, repeatedly applying the update \eqref{eq:jac-update} is a well known procedure often called ``twicing" \cite{imagefiltering}. By employing the decomposition of mean squared error (MSE) as the sum of the bias and variance of the estimator, the NTK analysis can be extended to the setting where the measurements $\mathbf{y}$ are corrupted by noise. The following theorem provides a formula for computing the MSE of image reconstruction with DIP in this setting, where the first term comes from the bias and the second from the variance.

\begin{theorem}[\textcolor{black}{Theorem~2 from \cite{liang2024analysis}}]
    Let $\mathbf{A}\in\mathbb{R}^{m\times n}$ be full row rank, and let $\mathbf{K} := \left( \nabla_\btheta f_\btheta(\z)\big\vert_{\btheta=\btheta_0} \right)\left( \nabla_\btheta f_\btheta(\z)\big\vert_{\btheta=\btheta_0} \right)^\top$ be the network's NTK. Suppose that the acquired measurements are $\mathbf{y} = \mathbf{Ax}+\mathbf{n}$, where $\mathbf{n} \sim \mathcal{N}(\boldsymbol{0}, \sigma^2\mathbf{I})$ and $\x \in \mathbb{R}^n$ is the true underlying signal. Then the MSE for DIP-based image reconstruction at iteration $t$ is given by:
    \[\textrm{MSE}_t = ||(\mathbf{I} - \eta \mathbf{K}\mathbf{A}^\top\mathbf{A})^t\mathbf{x}||_2^2 + \sigma^2 \sum_{i=1}^m \nu_{t,i}^2\:,\]
    where $\nu_{t, i}$ are the singular values of the matrix $(\mathbf{I} - (\mathbf{I} - \eta \mathbf{K}\mathbf{A}^\top\mathbf{A})^t)\mathbf{A}^\dagger$.
\end{theorem}

Finally, we empirically demonstrate that the analysis in the NTK regime \textcolor{black}{has the potential to} provide useful insights into the training dynamics of real networks. In Fig. \ref{fig:1d_experiment}, we show the results of denoising a 1D square signal using DIP with a 1D CNN. The network used is a 3-layer CNN with 256 hidden channels and ReLU activations. We also computed the empirical NTK $\mathbf{K}$ at initialization. We compare the peak-signal-to-noise ratio of the denoised signal obtained by the optimization \eqref{eqn: standard DIP} with the filtering in \eqref{eq:jac-update}. We find that closed-form filtering and real DIP show very similar behavior. \textcolor{black}{This result is interesting because it shows that the NTK perspective on DIP can effectively explain the phenomenon of early signal recovery followed by over-fitting, at least in this simple setting. We note that whether or not these conclusions extend to the networks typically used in DIP for 2D image reconstruction, such as U-Nets, is an important open question. Some notable results in this direction were obtained in \cite{tachella2021neural}, but the empirical investigation of analytical filtering with the NTK was limited to relatively simple network architectures.}

We also plot the singular values of $\mathbf{K}$, finding that the NTK's singular values decay quickly. Furthermore, $\mathbf{K}$ is poorly conditioned; even for this relatively small network the condition number of $\mathbf{K}$ is greater than $4 \times 10^3$. We also found that the condition number of the NTK grows quickly with network depth. For example, the NTK of a network with the same architecture but 15 hidden layers had a condition number greater than $10^{10}$. The fact that the NTK for deep networks is nearly singular has important \textcolor{black}{practical} implications. \textcolor{black}{In particular, it tells us that results (2) and (3) of Theorem \ref{thm:noiseless_recovery} may hold approximately for real image reconstruction with DIP using deep networks, since both of these results require the NTK to be singular. Moreover, it also suggests that the top singular vectors of $\mathbf{K}$ may dominate the reconstruction process in earlier iterations (perhaps promoting signal recovery). We would expect this because $\mathbf{K}$ is symmetric, so writing its singular value decomposition as $\mathbf{V}\mathbf{\Sigma}\mathbf{V}^\top$ shows that applying this matrix to a signal will tend to amplify frequencies or signal content aligned with the top singular vectors in $\mathbf{V}$, since the spectrum of $\mathbf{K}$ decays quickly.
On the other hand, after performing such a filtering many times, the trailing singular vectors may have a relatively larger effect on the gradient descent updates if the current reconstruction already approximately lies in the subspace spanned by the top singular vectors. This would lead one to anticipate successful early signal recovery, followed by performance degradation, depending on the alignment between the singular vectors of $\mathbf{K}$ and the true signal.}

For example, case (3) provides a condition for exact recovery under three conditions: (i) the NTK $\mathbf{K}$ is singular, (ii) $P_{N(\mathbf{A}) \cap R(\mathbf{K})} \mathbf{x} = \boldsymbol{0}$, and (iii) $\mathbf{x}\in R(\mathbf{K})$. Condition (i) indicates that the poor conditioning (or near low-rankness) of the NTK may be key to the success of DIP in image reconstruction. Condition (ii) is related to the information content of the measurements $\y$, or the ``incoherence" between the NTK and the forward operator, in the sense that the NTK does not readily produce signals in the null space of $\A$. Finally, condition (iii) relates to the ability of the NTK to represent the true signal.


\textcolor{black}{
Although NTK analysis reveals that gradient descent in the infinite-width limit inherently favors low-frequency components, this bias originates directly from the fixed kernel evaluated at initialization. While this allows for smooth approximations, NTK analysis critically overlooks the active role of over-parameterization effects introduced by network depth. More importantly, it does not promote feature learning, that is, the network weights stay too close to the initialization. In contrast, recent work on the implicit bias of gradient flow in two-layer matrix factorization addresses how depth influences image reconstruction by allowing for and promoting active  learning.}

\begin{figure}
    \centering
    \begin{tabular}{ccc}
        \includegraphics[width=0.3\linewidth]{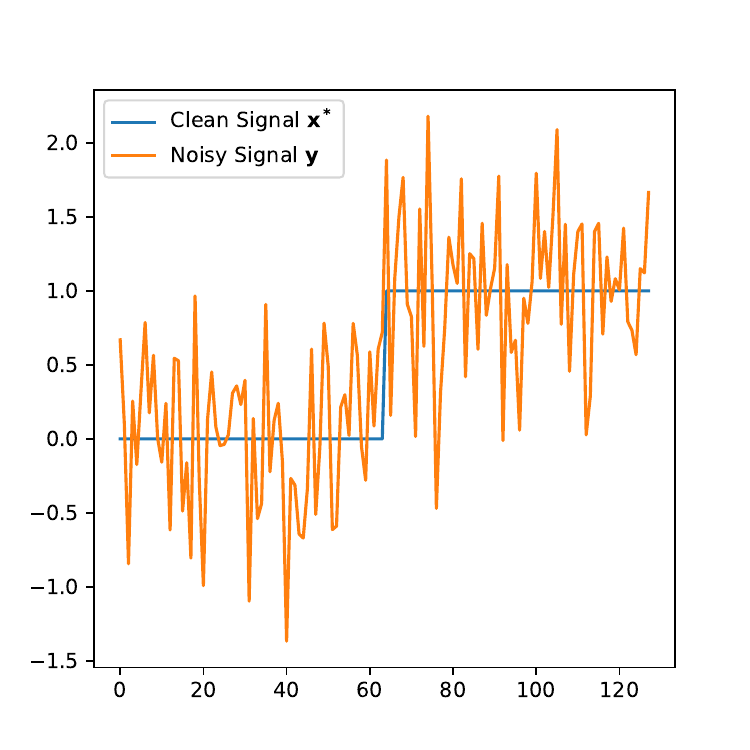}
        & \includegraphics[width=0.3\linewidth]{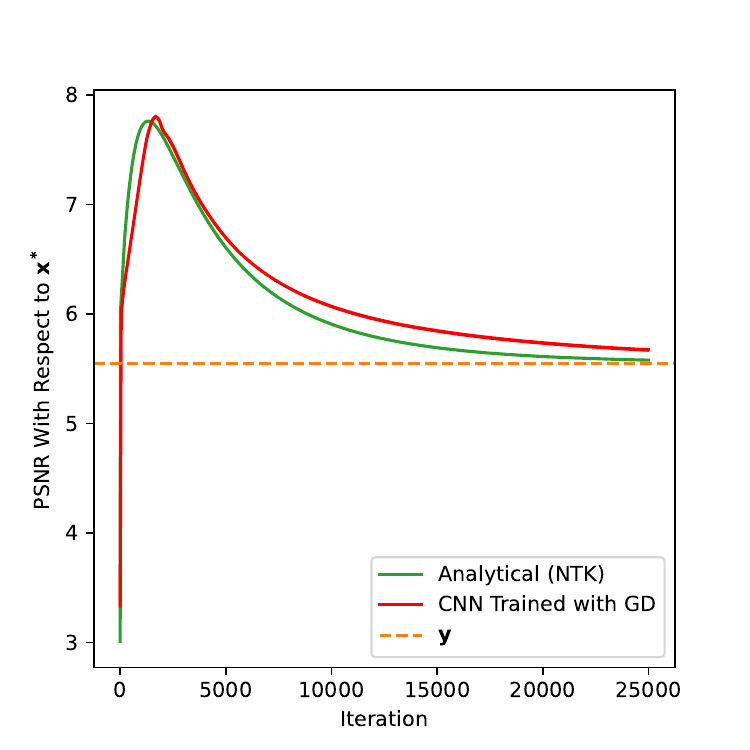}
        & \includegraphics[width=0.3\linewidth]{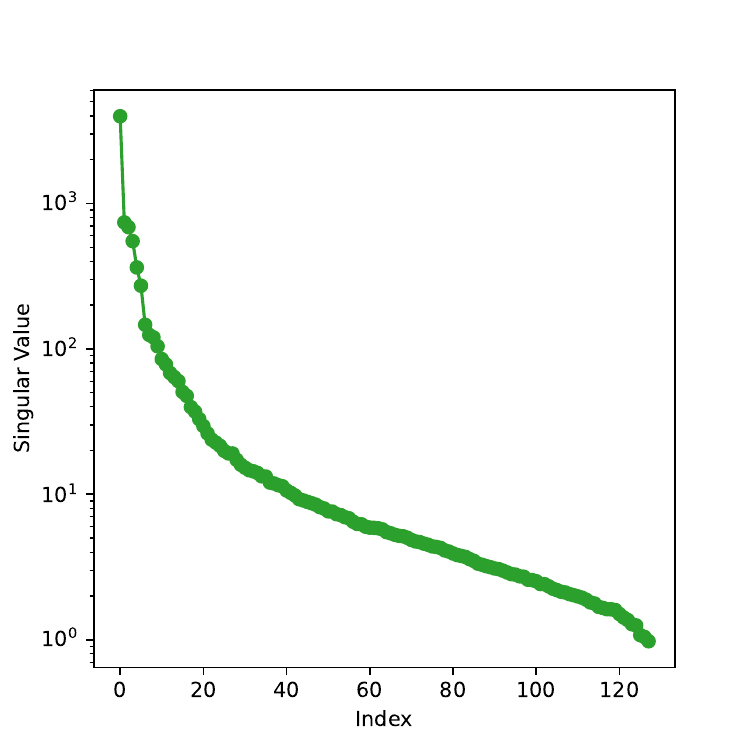}
    \end{tabular}
    \vspace{-0.4cm}
    \caption{1D signal denoising experiment using DIP. Left to right: (a) the clean and noisy signals used in the experiment, (b) PSNR of the denoised signal over iterations by training a 1D CNN with gradient descent vs. applying the update in equation \eqref{eq:jac-update} with its NTK, and (c) singular values of the NTK. The theoretical prediction aligns closely with the behavior of the real network.}
    \label{fig:1d_experiment}
    \vspace{-0.6cm}
\end{figure}


\subsubsection{\textcolor{black}{Burer-Monteiro Factorization and Two-Layer Matrix Factorization in Deep Image Prior}}

Another line of inquiry interprets DIP through the lens of \textcolor{black}{over-parameterized} matrix factorization, represented by works including \cite{ding2209validation,you2020robust}.
The deep image prior (DIP) framework is inherently \textcolor{black}{over-parameterized}, as the number of network parameters is significantly larger than the number of image pixels. Despite this \textcolor{black}{over-parameterization}, DIP exhibits a strong inductive bias towards natural images, preventing it from learning arbitrary noise. This phenomenon has been studied through the lens of implicit bias in optimization, particularly in \textcolor{black}{over-parameterized} models where gradient descent exhibits a preference for structured solutions. One such theoretical explanation comes from low-rank factorization methods, particularly the Burer-Monteiro (BM) factorization, which provides insights into how \textcolor{black}{over-parameterized} neural networks have an implicit bias towards low-rank solutions.

A common approach to studying implicit bias in \textcolor{black}{over-parameterized} optimization is through the matrix factorization model, where the image or signal is parameterized as $\mathbf{X} = \mathbf{U}\mathbf{U}^\top$, with $\mathbf{U} \in \mathbb{R}^{n \times r}$. This formulation replaces the explicit optimization over $\mathbf{X}$ with a factored representation, introducing \textcolor{black}{over-parameterization}. Given measurements $\mathbf{y} = \mathbf{A}(\mathbf{X}^*)$, where $\mathbf{X}^*$ is the ground-truth low-rank image, the problem is formulated as minimizing the least squares objective:
\begin{equation}
\label{BM-fit}
\min_{\mathbf{U}} \frac{1}{2m} \| \mathbf{A}(\mathbf{U}\mathbf{U}^\top) - \mathbf{y} \|_2^2, 
\end{equation}
where $\mathbf{A}: \mathbb{R}^{n \times n} \to \mathbb{R}^{m}$ is the measurement operator, and $m$ represents the number of compressive measurements. Here the linear measurements $\mathbf{y}=\{y_{i}\}_{i=1}^m$ are generated linearly as $y_{i} = \langle \mathbf{A}_i,\mathbf{X}^{*}\rangle$ for $i=1,2,...,m$. 

This formulation is \textcolor{black}{over-parameterized} when $r > \text{rank}(\mathbf{X}^*)$, meaning $\mathbf{U}$ has more columns than necessary. Despite this, gradient descent on this objective exhibits a strong implicit bias towards low-rank solutions. \textcolor{black}{The central idea is that the gradient flow dynamics, particularly when initialized with a very small initialization scale implicitly guide the solution towards a minimum nuclear norm solution. This is because the optimization trajectory, under certain conditions on the measurement operator $\mathbf{A}$, is constrained in a way that aligns with the Karush-Kuhn-Tucker (KKT) conditions of the nuclear norm minimization problem.}

\begin{theorem}[\textcolor{black}{Theorem~1 from \cite{gunasekar2017implicit}}]
\label{thm:commute}
Let $\mathbf{A}:\mathbb{R}^{n\times n}\to\mathbb{R}^m$ be a linear measurement operator defined by $(\mathbf{A}(\mathbf{X}))_i = \langle \mathbf{A}_i,\mathbf{X}\rangle$ for $i=1,\dots,m,$ where each $\mathbf{A}_i$ is a real symmetric matrix, and all $\mathbf{A}_i$ commute (that is, $\mathbf{A}_i\mathbf{A}_j = \mathbf{A}_j\mathbf{A}_i$ for every $i,j$). For a small scalar $\alpha>0$, define the scaled initialization $\mathbf{X}_{\alpha}(0) =\alpha\mathbf{X}_0$. Suppose that, starting from $\mathbf{X}_{\alpha}(0)$ and running gradient flow on loss \eqref{BM-fit} with $\mathbf{X}=\mathbf{U}\mathbf{U}^\top$, we converge (as $t\to\infty$) to a global minimizer $\mathbf{X}_{\alpha}(\infty)$. Assume there is a well-defined limit
$
\widehat{\mathbf{X}}
=
\lim_{\alpha\to 0}\,\mathbf{X}_{\alpha}(\infty)$
and 
$\mathbf{A}(\widehat{\mathbf{X}})=\mathbf{y}$.
Then $\widehat{\mathbf{X}}$ is a solution to the following convex problem:
$
\min_{\mathbf{X}\succeq \mathbf{0}}\;\|\mathbf{X}\|_*
\quad
\text{subject to}
\quad
\mathbf{A}(\mathbf{X})\;=\;\mathbf{y}.
$
\end{theorem}

\begin{proof}

We begin by observing that each matrix $\mathbf{A}_i$ is real and symmetric and that all $\mathbf{A}_i$ commute, meaning $\mathbf{A}_i\mathbf{A}_j = \mathbf{A}_j\mathbf{A}_i$ for every $i,j$. Real symmetric matrices are orthogonally diagonalizable, and commuting diagonalizable operators admit a single orthonormal basis in which they are all diagonal. Concretely, there is one basis $\{\mathbf{v}_1,\dots,\mathbf{v}_n\}$ that simultaneously diagonalizes every $\mathbf{A}_i$, so any linear combination $\mathbf{A}^*(\mathbf{r})
\;=\;
\sum_{i} r_i\,\mathbf{A}_i$ is also diagonalizable in that basis.
Let $\mathbf{r}_{t}$ denote the loss residual at each step given as $\mathbf{r}_{t}= \mathbf{A}(\mathbf{X}_{t})- \mathbf{y}$,
then the Gradient flow iterates on $\mathbf{U}_{t}$  are given as:
\begin{equation}
\dot{\mathbf{U}}_{t} = - \mathbf{A}^*\big(\mathbf{A}(\mathbf{U}_{t}\mathbf{U}_{t}^{T}) - \mathbf{y}\big) \mathbf{U}_{t} = - \mathbf{A}^*(\mathbf{r}_{t}) \mathbf{U}_{t}
\end{equation}
This dynamics defines the behaviour of $\mathbf{X}_{t} =\mathbf{U}_{t}\mathbf{U}^{T}_{t} $ and using the chain rule, we get 
\begin{equation}
\dot{\mathbf{X}}_{t} = \dot{\mathbf{U}_{t}}\mathbf{U}^{T}_{t} + \mathbf{U}_{t} \dot{\mathbf{U}^{T}_{t}} = -\mathbf{A}^*(\mathbf{r}_t)\,\mathbf{X}_t - \mathbf{X}_t\,\mathbf{A}^*(\mathbf{r}_t).
\end{equation}
This equation has an explicit closed form solution given as:
\begin{equation}
\label{GF-traj}
\mathbf{X}_t = \exp(\mathbf{A}^{*}(\mathbf{s}_t)) \mathbf{X}_0 \exp(\mathbf{A}^{*}(\mathbf{s}_t)),
\end{equation}
where $\mathbf{s}_T = -\int_0^\top \mathbf{r}_t dt$. Assuming that the solution with a very small initialization $\alpha$ , i.e., $\hat{\mathbf{X}} = \lim_{\alpha \rightarrow 0} \mathbf{X}_{\alpha}(\infty)$ exists and $\mathbf{A}(\hat{\mathbf{X}})= \mathbf{y} $, when we converge to the zero global minima\footnote{The loss landscape in \eqref{BM-fit} is benign, i.e., it only consists of saddles and global minima and no local minima. So gradient flow is guaranteed to converge to a global minimum $\mathbf{A}(\hat{\mathbf{X}})= \mathbf{y}$.}, we want to show that $\hat{\mathbf{X}}$ is the solution to the problem $\min_{\mathbf{X} \succeq \mathbf{0}} \|\mathbf{X}\|_* \quad \text{subject to} \quad \mathbf{A}(\mathbf{X}) = \mathbf{y}.$
The KKT optimality conditions for the above optimization problem are: 
\begin{equation}
 \quad \mathbf{A}(\hat{\mathbf{X}}) = \mathbf{y}, \quad \hat{\mathbf{X}} \succeq \mathbf{0}, \quad \text{and}  \quad \mathbf{A}^{*} (\boldsymbol{\nu}) \preceq \mathbf{I}, \quad (\mathbf{I} - \mathbf{A}^{*} (\boldsymbol{\nu})) \hat{\mathbf{X}} = \mathbf{0} \quad \text{for some}\quad  \boldsymbol{\nu} \in \mathbb{R}^m.
\end{equation}
We already know the first condition holds as it is the global minimizer, the positive semidefiniteness condition ($\hat{\mathbf{X}} \succeq \mathbf{0}$) is ensured by the factorization $\mathbf{X} = \mathbf{U}\mathbf{U}^\top$. The remaining complementary slackness and dual feasibility conditions effectively require that $\hat{\mathbf{X}}$ be spanned by the top eigenvector(s) of $\mathbf{A}$. From the dual feasibility condition $\mathbf{A}^*(\boldsymbol{\nu}) \preceq \mathbf{I}$,
we know that all eigenvalues of $\mathbf{A}^*(\boldsymbol{\nu})$ are at most 1. The complementary slackness condition $(I - \mathbf{A}^{*}(\boldsymbol{\nu}))\,\mathbf{X} = \mathbf{0}$ then forces $\mathbf{X}$ to lie in the subspace where $\mathbf{A}^*(\nu)$ has eigenvalues 1 (i.e., its top eigenvalue/eigenvectors). Consequently, $\mathbf{X}$ is spanned by only those top eigenvectors of $\mathbf{A}^{*}$ 
where the eigenvalue is 1.
The gradient flow (GF) trajectory in equation \eqref{GF-traj}, for any non-zero $\mathbf{y}$ satisfies these KKT conditions as the initialization scale $\alpha \rightarrow 0$. As $\alpha \rightarrow 0$, the solution path of the trajectory $\mathbf{X}_{\alpha}(t) = \exp(\mathbf{A}^{*}(\mathbf{s}_t)) \mathbf{X}_{\alpha}(0) \exp(\mathbf{A}^{*}(\mathbf{s}_t))$ has the following properties. Let $\lambda_{k}(\mathbf{X}_{\alpha}(\infty))$ denote the $k^{th}$ eigenvalue of the limiting solution $\mathbf{X}_{\alpha}(\infty)$, then it can be shown that for all $k$ such that $\lambda_{k}(\mathbf{X}_{\alpha}(\infty)) >0 $, we have:
\begin{align}
\label{eig-infty}
\lambda_{k}\left(\A^*\!\left(\frac{\mathbf{s}_{\infty}(\beta)}{\beta}\right)\right)
- 1
- \frac{\ln\!\big(\lambda_{k}({\mathbf{X}}_{\alpha}(\infty))\big)}{2\,\beta}
\;\rightarrow\;0\:,
\end{align}
where $\beta = - \log(\alpha)$. \eqref{eig-infty} is obtained by comparing the $k^{th}$ eigenvalues of the trajectory $\mathbf{X}_{\alpha}(t)$. Defining $\mathbf{\gamma}(\beta)=\frac{\mathbf{s}_{\infty}(\beta)}{\beta}$, it can be concluded that for all $k$ if $\lambda_{k}(\hat{\mathbf{X}}_{\alpha}(\infty)) \neq 0$, then $\lim_{\beta \rightarrow \infty} \lambda_{k}(\A^*\!(\gamma(\beta))) =1$. Similarly for each $k$ such that $\lambda_{k}(\hat{\mathbf{X}}_{\alpha}(\infty))=0$, we obtain 
$\exp\big(\lambda_k\big(\mathbf{A}^*(\nu(\beta))\big) - 1\big)^{2\beta} \;\to\; 0\:,$
for large $\beta$, this implies $\lambda_k(\mathbf{A}^*(\nu(\beta)) <1$. Since, we denoted $\hat{\mathbf{X}} = \lim_{\alpha \rightarrow 0} \mathbf{X}_{\alpha}(\infty)$ as the limiting solution at very small initialization, we have $\lim_{\beta \rightarrow \infty} \A^*\!(\gamma(\beta)) \preceq \mathbf{I}$ and $\lim_{\beta \rightarrow \infty} \A^*\!(\gamma(\beta))\hat{\mathbf{X}}=\hat{\mathbf{X}} $. 
\end{proof}

This theorem formally establishes why gradient flow on the factorized representation $\mathbf{X} = \mathbf{U}\mathbf{U}^\top$ converges to the minimum nuclear norm solution. Essentially, the dynamics of gradient flow ensure that the optimization remains constrained to a low-rank subspace dictated by the top eigenvectors of $\mathbf{A}$. This means that this \textcolor{black}{over-parameterization} implicitly regularizes the solution by preferring low-rank structures.
This result offers a rigorous theoretical foundation for implicit regularization in deep image prior (DIP) models through \textcolor{black}{over-parameterized} factorization. \textcolor{black}{To make the presentation simple, measurement noise  is not assumed in Theorem~\ref{thm:commute} and the measurements are directly generated from the ground-truth as $\mathbf{y} = \mathbf{A}(\mathbf{X}^*)$. However, this can be extended to include measurement noise, as shown in Theorem 2.3 of \cite{ding2022validation}, where the authors demonstrate that low-rank solutions are recovered earlier, before overfitting occurs, in the presence of noise.}

The fact that gradient descent naturally avoids fitting measurement noise during the initial optimization iterations of DIP is intimately linked with an implicit nuclear norm minimization process.
Specifically, as the evolution of 
$\mathbf{X}_t = \mathbf{U}_t \mathbf{U}_t^\top$
keeps the optimization trajectory within a low-rank subspace, it biases the solution toward structured, low-complexity representations. Although in practice, neural networks have more complicated dynamics because they involve several non-linearities, the above theorem gives a simple example on how \textcolor{black}{over-parameterization} can implicitly bias the network \textcolor{black}{output trajectory to low-rank solutions}. However, $\mathbf{X}=\mathbf{U}\mathbf{U}^{T}$ usually assumes that the image $\mathbf{X}$ can be expressed as an output of symmetric encoder decoder type network. However, similar analysis with reparameterization $\mathbf{X}=\mathbf{U}\mathbf{V}^{T}$ may also yield a low-rank bias. 
\textcolor{black}{Although DIP initially shows a strong initial implicit bias toward low-rank solutions due to over-parameterization, prolonged training leads to the loss reaching zero, perfectly fitting the corrupted measurements. This has prompted extensive research into methods for avoiding overfitting, which we'll discuss further in the next section.}

\vspace{-0.2in}

\section{Recent Algorithms for Addressing the \textcolor{black}{Over-fitting} Issue}

This section introduces \textit{recent} DIP methods, categorized based on their approach to mitigating noise \textcolor{black}{over-fitting}. The first three subsections cover regularization techniques, early stopping strategies, and network re-parameterization methods, respectively. The final subsection discusses approaches that integrate DIP with pre-trained models. While numerous DIP variants exist, we focus on those that meet the following criteria: applicability to a range of tasks, competitive quantitative performance, and demonstrated robustness against noise \textcolor{black}{over-fitting}.

\vspace{-0.2in}

\subsection{Regularization-based Methods}


\subsubsection{Self-Guided DIP}
In Self-Guided DIP~\cite{liang2024analysis}, the authors proposed a denoising regularization term along with optimizing over the input and the parameters of the network. Specifically, the authors proposed
\begin{equation}\label{eqn: self-G DIP}
    \boldsymbol{\theta}',\mathbf{z}' = \arg \min_{\boldsymbol{\theta},\mathbf{z}} \|\mathbf{A}\mathbb{E}_{\boldsymbol{\eta}}[f_{\boldsymbol{\theta}}(\mathbf{z}+{\boldsymbol{\eta}})]-\mathbf{y}\|_2^2+ \lambda \|\mathbb{E}_{\boldsymbol{\eta}}[f_{\boldsymbol{\theta}}(\mathbf{z}+{\boldsymbol{\eta}})]-\mathbf{z}\|_2^2\:,
\end{equation}
where ${\boldsymbol{\eta}}$ is a random noise vector drawn from some distribution (either uniform or Gaussian). The final reconstruction is obtained as $\hat{\mathbf{x}} = \mathbb{E}_{\boldsymbol{\eta}}[f_{\boldsymbol{\theta}'}(\mathbf{z}'+{\boldsymbol{\eta}})]\:,$ where the expectation is replaced in implementations by an average of the network outputs over a fixed number of random input perturbations with noise.
The motivation of Self-Guided DIP is to remove the prior data dependence in Reference-Guided DIP \cite{zhao2020reference} (Ref-G DIP) where the authors have shown that using a prior image as input to the network (i.e., $\mathbf{z}$) improves performance. 
The regularization exploiting purely synthetic noise (and denoising) also plays a key role in performance. 
In addition to the regularization parameter, $\lambda$, the selection of ${\boldsymbol{\eta}}$ and the implementation of the expectation is also considered a hyperparameter in Self-Guided DIP.

Self-Guided DIP was evaluated on MRI reconstruction and inpainting tasks. Notably, for MRI, across various acceleration factors, modalities, and datasets, Self-Guided DIP not only outperformed Ref-G DIP (a method that requires a prior image) but also demonstrated that a DIP-based approach could surpass the well-trained supervised model such as MoDL~\cite{8434321}.

\subsubsection{Autoencoding Sequential DIP}
The authors in \cite{alkhouriNeuIPS24} introduced autoencoding Sequential DIP (aSeqDIP) which uses an autoencoding regularization term in addition to an input-adaptive objective function. Specifically, the updates in the aSeqDIP algorithm are 
\begin{equation}\label{eqn: aseqdip}
    \boldsymbol{\theta} \leftarrow \arg \min_{\boldsymbol{\theta}} \|\mathbf{A}f_{\boldsymbol{\theta}}(\mathbf{z}) - \mathbf{y}\|^2_2 + \lambda \|f_{\boldsymbol{\theta}}(\mathbf{z})-\mathbf{z}\|^2_2\:, \quad\quad \mathbf{z} \leftarrow f_{\boldsymbol{\theta}}(\mathbf{z})\:.
\end{equation}
The optimization in \eqref{eqn: aseqdip} is run for a few gradient steps and corresponds to the network parameters' update whereas the second part represents the network input update. These updates are run for the same number of optimization steps in Vanilla DIP. The hyperparameters in aSeqDIP are the regularization parameter and the number of updates in the second part of \eqref{eqn: aseqdip}. 

The motivation/intuition of aSeqDIP is the impact of the DIP network input on the performance. While the authors in \cite{tachella2021neural} have considered how a structured DIP network input can impact performance, the authors of aSeqDIP explored the employing a noisy version of the ground truth as the fixed input to the DIP objective in \eqref{eqn: standard DIP}. In particular, it was empirically shown that a closer similarity of DIP network input to the ground truth (i.e., less noise) corresponds to higher reconstruction quality. This led to the development of the input-adaptive algorithm that is based on the updates in \eqref{eqn: aseqdip}. 


Theoretically, the authors in aSeqDIP showed the impact of the DIP input through an NTK study using CNNs with a residual connection (Theorem A.1 in \cite{alkhouriNeuIPS24}). Empirically, aSeqDIP was evaluated on two medical image reconstruction tasks (MRI and sparse view CT) and three image restoration tasks (denoising, inpainting, and non-linear deblurring). Notable empirical results are: (\textit{i}) aSeqDIP was shown to have higher resilience to noise \textcolor{black}{over-fitting} when compared to other regularization-based DIP methods (see the results in Section~\ref{sec: exp DIP comparison}), and (\textit{ii}) quantitatively, aSeqDIP was shown to be either on-par with or outperform data-centric diffusion-based generative methods, such as Score-MRI \cite{chung2022score} and DPS \cite{chung2022diffusion}, that use models pre-trained with an extensive amount of data.   

\vspace{-0.2in}
\subsection{Early Stopping DIP}











The authors in \cite{wang2021early} proposed ES-DIP which estimates the optimization iteration for near-peak DIP PSNR performance by computing the running variance of intermediate reconstructions. The motivation of ES-DIP lies in the relation between the peak in the PSNR curve and the minimum of the moving variance curve observed in vanilla DIP. Then, based on this observation, the authors try to obtain the windowed moving variance (WMV). 

Specifically, let $W$ be the time window size and $P$ be the patience (duration or number of iterations) for which the variance does not change significantly. Then, the ES-DIP algorithm computes
\begin{equation}\label{eqn: ES-DIP}
    \textrm{VAR}_t \doteq \frac{1}{W} \sum^{W-1}_{w=0} \Big\| \mathbf{x}^{t+w} - \frac{1}{W} \sum^{W-1}_{i=0} \mathbf{x}^{t+i} \Big\|^2_2\:.
\end{equation}
If this value does not significantly change for $P$ iterations (the patience), the ES takes place. This means that $W$ and $P$ represent the main hyperparameters in ES-DIP.

Theoretically, using the DIP NTK approximation in \eqref{eqn: NTK first eqn} along with the window size $W$ shows how $\textrm{VAR}_t$ in \eqref{eqn: ES-DIP} depends on the singular values and left singular vectors of $\mathbf{J}$ (Theorem~2.1 in \cite{wang2021early}). The main insight of this theorem is that when the learning rate is sufficiently small, the WMV of $\mathbf{x}_t$ decreases monotonically. On this basis, the authors derived an upper bound of the WMV of $\mathbf{x}_t$ that depend on many parameters including the window size $W$, singular values of $\mathbf{J}$, and the measurements $\mathbf{y}$ (see Equation~(7) in the statement of Theorem~2.2 in \cite{wang2021early}). 

ES-DIP was evaluated for the tasks of denoising, super resolution, and MRI reconstruction. The authors have considered different noise types and levels. It was also extended to the blind setting by considering the blind image deblurring task. In addition to achieving good quantitative results, ES-DIP's run-time is also faster than many other DIP-based methods (\textcolor{black}{as will be demonstrated in Table~\ref{tab:run-time} of Section~\ref{sec: exp DIP comparison}}). The authors have shown that the proposed stopping criteria can be used to improve the performance of other network structures such as the network under-parameterized architecture in Deep Decoder \cite{heckel2019deep} (see the results of Fig. 10 in \cite{wang2021early}). 

\vspace{-0.2in}
\subsection{Network Re-parameterization Methods}
\subsubsection{Deep Decoder}

\textcolor{black}{Motivated by the need for image priors that avoid overfitting in inverse problems, Deep Decoder, proposed by \cite{heckel2019deep}, is an under-parameterized neural network that solely consists of the decoder portion of a U-Net.} Unlike traditional neural networks that rely on deep convolutional layers, the Deep Decoder avoids convolutional layers entirely and instead leverages upsampling operations, pixel-wise linear combinations, ReLU activations, and channel-wise normalization. This under-parameterization acts as an implicit regularizer, allowing the network to perform effectively in tasks such as image denoising without requiring explicit training. 

The Deep Decoder architecture follows a simple repetitive pattern of operations: a $1 \times 1$ convolution layer followed by an upsampling operation, a ReLU non-linearity, and a channel-wise normalization step. The standard Deep Decoder configuration uses six layers with a channel dimension of 128, resulting in approximately 100,224 parameters, significantly fewer than the number of pixels in an RGB image of size $512 \times 512$ (786,432 pixels). This under-parameterization serves as a natural regularizer, making the Deep Decoder robust to \textcolor{black}{over-fitting}. 

In terms of practical applications, Deep Decoder has been demonstrated to perform well in denoising and image reconstruction tasks.
Despite having fewer parameters than traditional models, it achieves denoising quality comparable to wavelet-based compression methods. Its effectiveness stems from the network's inherent structure, which biases it toward generating natural images even without explicit training.


\subsubsection{Optimal Eye Surgeon (OES)}

While Deep Decoder provides a strong baseline for under-parameterized networks, it is limited to a fixed decoder architecture. Optimal Eye Surgeon (OES) generalizes this concept to a wider class of neural architectures. Instead of strictly relying on a predefined structure, OES enables a principled pruning approach at initialization, effectively learning sparse sub-networks from over-parameterized models and then training these sub-networks to reconstruct images. In particular, OES first adaptively prunes the network at initialization so the mask it has learned is optimized to the underlying image or its measurements, then this subnetwork weights are updated to fit the measurement.

The working principle of OES is also based on under-parameterization just like the Deep Decoder that ensures the recovered image does not overfit to the target and corrupted measurements. OES generally consists of two stages: (\textit{i}) for a user specified sparsity level, the binary mask is learned using Gumbel Softmax re-parameterization which learns a Bernoulli distribution over the parameter space. These probabilities denote the importance of each parameter in generating the underlying measurement. (\textit{ii}) This learned mask is applied to the network to obtain a subnetwork which is trained to fit the image. 

Specifically, given a random initialization for parameter $\boldsymbol{\theta}_{in}$, OES finds a binary mask from the underlying measurements, i.e., $\mathbf{m}^*(\mathbf{y})$ with a given user-defined sparsity $\|\mathbf{m} \|_{0}<s$. Since solving a discrete optimization problem for neural networks is challenging, the authors \cite{ghosh2024optimal} propose a Bayesian relaxation (shown in Equation \eqref{eq:ber} below). This optimization problem is unconstrained, continuously differentiable and can be solved by iterative algorithms such as Gradient Descent after proper re-parameterization using the Gumbel Softmax Trick. Instead of learning the binary mask, the assumption is that the mask is sampled from a Bernoulli distribution $\mathbf{m} \sim Ber(\mathbf{p})$ and learn the probabilities of the mask $\mathbf{p}$ instead. The sparsity constraint is implemented through the KL divergence regularization which ensures that the learned sparsity level arises from $\mathbf{p}$ being close to a user defined $\mathbf{p}_{0} = \frac{s}{d}$, where $d$ is the parameter dimension. The probabilities corresponding to the parameters, i.e ., $\mathbf{p}$, are learned through the Gumbel Softmax trick. After learning $\mathbf{p}$, the weights are pruned based on the larger magnitudes of $\mathbf{p}$ to reach the desired sparsity level, given by the threshold function $C$.  Let $G(\boldsymbol{\theta}_{in} \textcolor{black}{\odot} \mathbf{m}, \mathbf{z})$ denote the image generator network initialized with random weights $\boldsymbol{\theta}_{in}$, and $\lambda$ denotes the regularization strength for the KL term. Then the mask learning optimization problem can be formulated as follows:
\begin{equation}
\label{eq:ber}
    \mathbf{m}^*(\mathbf{y}) = C(\mathbf{p}^*) \quad \text{s.t.} \quad 
    \mathbf{p}^* = \arg \min_{\mathbf{p}} \mathbb{E}_{\mathbf{m} \sim \text{Ber}(\mathbf{p})} 
    \left[ \| \mathbf{A}G(\boldsymbol{\theta}_{in} \odot \mathbf{m}, \mathbf{z}) - \mathbf{y} \|_2^2 \right] 
    + \lambda KL\left(\text{Ber}(\mathbf{p}) \| \text{Ber}(\mathbf{p}_0)\right).
\end{equation}
%

A key advantage of OES is its flexibility: it can be applied to any deep convolutional network architecture, making it broadly useful across different inverse problems. Experiments in \cite{ghosh2024optimal} demonstrate that OES-based subnetworks surpass other state-of-the-art pruning strategies such as the Lottery Ticket Hypothesis (which prunes based on the magnitude of the weights at convergence) in image denoising and recovery tasks. Another key advantage of OES is that a mask can be learned from one image as the target and then that masked subnetwork can be trained for denoising a different image. This is particularly useful when an image dataset includes images from diverse classes. 



\subsubsection{Double Over-parameterization (DOP)}

While Deep Decoder and OES use fewer parameters than the standard DIP neural network architecture, the authors in \cite{you2020robust} introduced DOP, a method that embraces over-parameterization for measurement noise modeling. This introduced over-parameterization imposes implicit regularization through the use of different learning rates for different components of the model. In particular, the authors introduced Hadamard-product-based over-parameterization \((\mathbf{g} \odot \mathbf{g} - \mathbf{h} \odot \mathbf{h})\), which, when optimized with small initialization and infinitesimal learning rate, effectively filters out the sparse noise in the measurements. The optimization problem in DOP is $\min_{\boldsymbol{\theta},\,\mathbf{g},\,\mathbf{h}} 
\left\| \mathcal{A}\,f_{\boldsymbol{\theta}}(\mathbf{z}) \;+\; 
(\mathbf{g} \odot \mathbf{g} \;-\; \mathbf{h} \odot \mathbf{h})
\;-\;
\mathbf{y} \right\|_{2}^{2}.$
As observed, variables \(\mathbf{g}\) and \(\mathbf{h}\) are introduced for modeling the noise in \(\mathbf{y}\). Theoretically, unlike previous approaches that require early stopping to prevent \textcolor{black}{over-fitting}, DOP's implicit bias ensures that no explicit stopping criterion is required. Along with the network's implicit bias towards natural images, the implicit bias of the Hadamard product captures the sparse noise.

A key insight into this implicit bias can be gained by considering a low-rank matrix factorization version of the problem, where the underlying variable is low-rank and is represented as \(\mathbf{U}\mathbf{U}^\top\) instead of \(f_{\boldsymbol{\theta}}(\mathbf{z})\). In this setting, the loss becomes $\min_{\mathbf{U},\,\mathbf{g},\,\mathbf{h}} \|\,
\mathcal{A}(\mathbf{U}\mathbf{U}^\top)
\;+\;
(\mathbf{g} \odot \mathbf{g} \;-\; \mathbf{h} \odot \mathbf{h})
\;-\;
\mathbf{y}
\|_{2}^{2}.$
With small initialization, gradient descent implicitly biases $f_{\boldsymbol{\theta}}(\mathbf{z})$ toward a low-nuclear-norm solution (i.e., enforcing low-rank) (refer to Theorem~4), while \(\mathbf{g} \odot \mathbf{g} - \mathbf{h} \odot \mathbf{h}\) remains sparse, mirroring an \(\ell_1\)-type penalty to capture the sparse noise. A central theoretical result in \cite{you2020robust} clarifies that this over-parameterized formulation, $\mathbf{X} \;=\; \mathbf{U} \mathbf{U}^\top, \mathbf{s}\;=\; \mathbf{g}\!\odot\!\mathbf{g} \;-\; \mathbf{h}\!\odot\!\mathbf{h},$
together with \emph{discrepant learning rates} for \(\{\mathbf{U}\}\) versus \(\{\mathbf{g}, \mathbf{h}\}\), yields a solution \((\mathbf{X}, \mathbf{s})\) that also solves the \emph{convex} program
\begin{equation}\label{eqn:implicit-bias}
\min_{\mathbf{X}\in\mathbb{R}^{n\times n},\, \mathbf{s}\in\mathbb{R}^m}
\;\;\| \mathbf{X}\|_* \;+\;\lambda\,\| \mathbf{s}\|_1
\quad
\text{subject to} \\
\quad
\mathcal{A}(\mathbf{X}) + \mathbf{s} = \mathbf{y},
\quad
\mathbf{X} \succeq \mathbf{0},
\end{equation}
where \(\lambda = 1/\alpha\). In other words, the ratio \(\alpha\) of the step sizes for \(\{\mathbf{g}, \mathbf{h}\}\) to that of \(\mathbf{U}\) acts as an implicit regularization parameter, balancing the nuclear norm of \(\mathbf{X}\) against the \(\ell_1\)-norm of \(\mathbf{s}\). By simply tuning \(\alpha\), one obtains the same trade-off that would otherwise require an explicit penalty \(\lambda\) in \(\|\mathbf{X}\|_* + \lambda \|\mathbf{s}\|_1\). Empirical results show that DOP provides superior performance compared to vanilla DIP model for image denoising, and it also outperforms traditional nuclear norm minimization techniques in low-rank matrix recovery. However, DOP’s increased parameterization can lead to higher computational cost relative to other methods.



\subsubsection{Deep Random Projector}

The work in \cite{Li_2023_CVPR} introduced the deep random projector (DRP), a method that combines three previously-explored approaches, and is proposed to mitigate the slowness issue (as we need a separate optimization for each measurement in DIP) in addition to the noise \textcolor{black}{over-fitting} problem. DRP proposes three modifications: (\textit{i}) optimizing over the input of the network and a subset of the network parameters, namely the batch normalization layers weights; (\textit{ii}) reducing the network depth (the number of layers); and (\textit{iii}) using the Total Variation (TV) prior regularization. In particular, the optimization takes place as $\min_{\mathbf{z}, \boldsymbol{\theta}_{\textrm{BN}}\subset \boldsymbol{\theta}} \| \mathbf{A}f'_{\boldsymbol{\theta}}(\mathbf{z}) - \mathbf{y}\|_2^2 + \lambda \rho_{\textrm{TV}}(f'_{\boldsymbol{\theta}}(\mathbf{z}))\:,$ 
where $\boldsymbol{\theta}_{\textrm{BN}}$ represents the affine parameters for batch normalization, and $f'$ represents the reduced depth network with the first layer as a batch normalization layer. In other words, $f'$ is a modified version of the standard network architecture which we defined earlier as $f$. Here, the second term is the total variation 
$\rho_\textrm{TV} = \lambda \sum^{n}_{i=1} | (\mathbf{D}_1f_{\boldsymbol{\theta}}(\mathbf{z}))_i| + | (\mathbf{D}_2f_{\boldsymbol{\theta}}(\mathbf{z}))_i|\:,$ 
where $\mathbf{D}_1$ and $\mathbf{D}_2$ are the finite difference operators for the first and second dimensions, respectively. The use of classical explicit TV regularization with DIP to mitigate noise \textcolor{black}{over-fitting} was explored in an earlier method, TV-DIP \cite{liu2019image}.


In addition to the regularization parameter, the number of reduced layers is also considered as a hyperparameter. DRP has been applied for tasks such as denoising, super-resolution, and inpainting. Empirically, DRP was shown to operate with the standard DIP network architecture as well as the deep decoder network architecture (described earlier in this subsection). Furthermore, DRP achieves notable speedups when compared to other methods.








\subsection{Combining DIP with Other Pre-trained Models}\label{sec: combining DIP with other models}


This subsection explores recent studies that integrate DIP with pre-trained models. Specifically, these works investigate whether DIP can enhance the performance of an existing pre-trained model for a given IIP or whether a hybrid approach can combine the strengths of both. To this end, we discuss four hybrid methods, including two from the final category in Fig.~\ref{fig: timeline} (DeepRED and uDiG-DIP).

First is DeepRED \cite{Mataev_2019_ICCV}, which proposed to combine DIP with a pre-trained denoiser. The authors use the concept of Regularization by Denoising (RED)~\cite{romano2017little}, which leverages existing pre-trained denoisers, as an explicit regularization prior to improve the performance of vanilla DIP in terms of noise \textcolor{black}{over-fitting} mitigation. This work was evaluated on three image restoration tasks (denoising, super resolution, and deblurring) and was shown to outperform the standalone conventional RED \cite{romano2017little}. 




More recently, the DIP framework was combined with diffusion models (DMs) \cite{chung2022diffusion}, as presented in deep diffusion image prior (DDIP) \cite{ddip}, the sequential Diffusion-Guided DIP (uDiG-DIP) \cite{LianguDigDIP}, and the Constrained Diffusion Deep Image Prior (CDDIP) \cite{goyes2024cddip}. In DDIP, the authors propose using the DIP framework to enhance the out-of-distribution adaptation of DM-based 3D reconstruction solvers in a meta-learning framework where fine-tuning the weights of the DM is needed. On the other hand, uDiG-DIP \cite{LianguDigDIP}, inspired by the impact of the DIP network input (similar to aSeqDIP \cite{alkhouriNeuIPS24}), uses the DM as a diffusion purifier where at each gradient update (i.e., the second update of \eqref{eqn: aseqdip}), the DM is used to refine the network input. uDiG-DIP was applied to MRI and sparse view CT and was shown to achieve high robustness to noise \textcolor{black}{over-fitting} in addition to outperforming DM-only methods such as \cite{chung2022score,chung2022improving}. CDDIP employs Tweedie's formula \cite{chung2022diffusion} to estimate an image via a pre-trained DM. At each sampling step, this estimate serves as the DIP network's input, used to enforce measurement consistency. Applied to seismic reconstruction, CDDIP was shown to outperform standalone DMs in terms of reconstruction quality and reduced sampling time steps.






%
\begin{figure}[t]
\centering
\includegraphics[width=1\linewidth]{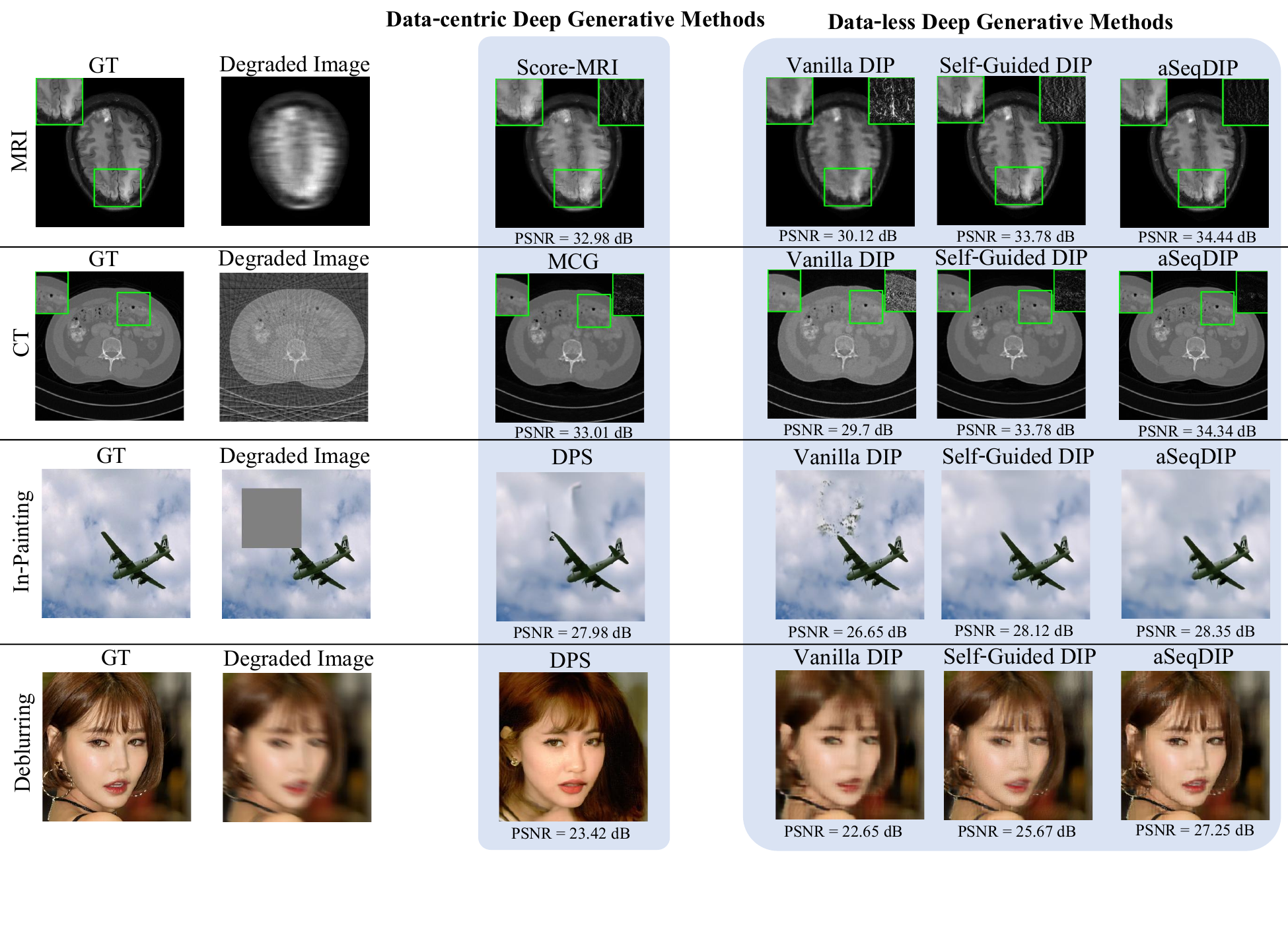}
\vspace{-2.122cm}
\caption{Reconstructed/recovered images using DM-based methods (3rd column) and data-less methods (columns 4 to 6). The ground truth (GT) and degraded images (under sampled measurements in MRI and CT, and corrupted images for natural image restoration) are shown in the first and second columns, respectively. PSNR results are given at the bottom of each reconstructed image. For MRI (4x undersampling) and CT (18 views), the top right box shows the absolute difference between the center region box of the reconstructed image and the same region in the GT image. For in-painting, we used hole to image ratio of $0.25$. For data-centric generative methods, we use Score-MRI \cite{chung2022score}, Manifold Constrained Gradient (MCG) \cite{chung2022improving}, and DPS \cite{chung2022diffusion}. \textcolor{black}{For Deblurring, aSeqDIP and self-guided DIP contain artifacts (e.g., the region near the ear) when compared to DPS (a data-centric method). However, DPS outputs a perceptually different image as compared to the GT. For MRI, CT, and box in-painting, aSeqDIP and self-guided DIP reconstructions contain sharper and clearer image features than other methods.} The images and settings of these experiments are sourced from Fig.~5 in \cite{alkhouriNeuIPS24}.}
\label{fig: visual}
\vspace{-0.4cm}
\end{figure}

\begin{figure}[t]
\centering
\includegraphics[width=0.8\linewidth]{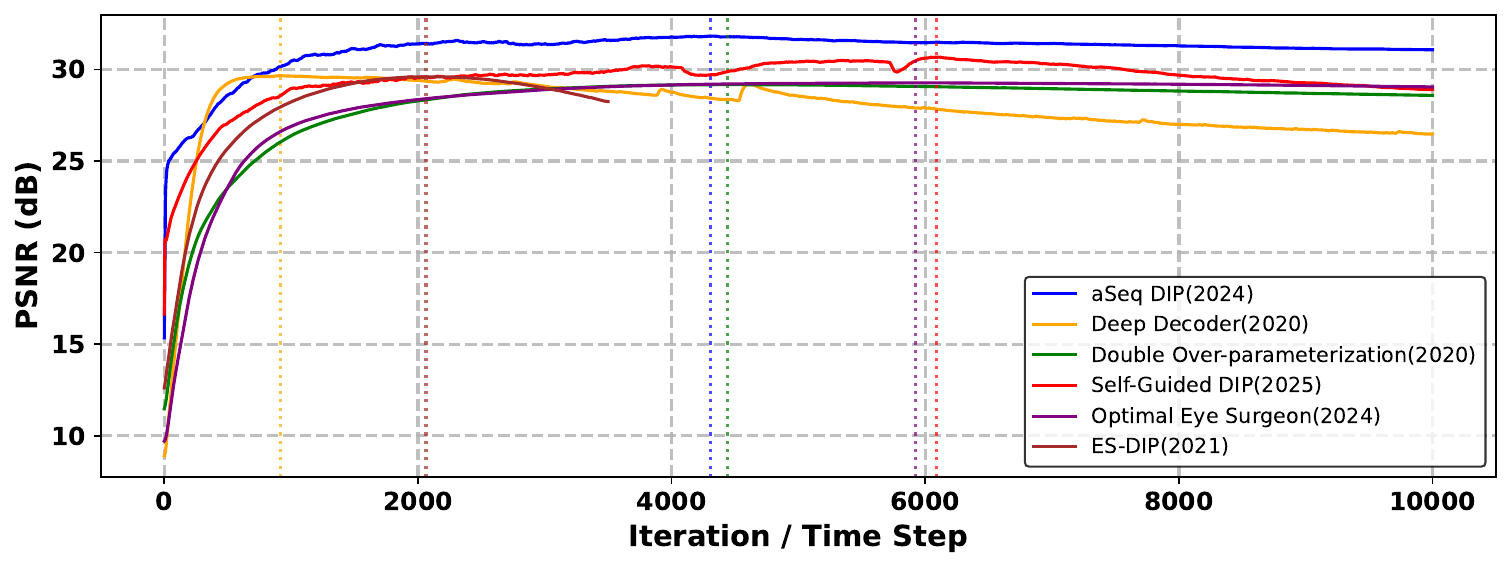}
\vspace{-0.6cm}
\caption{\textcolor{black}{Average PSNR curves (y-axis) for ten ImageNet test images using recent DIP-based methods for reducing noise over-fitting plotted over the optimization iteration/time-step (x-axis). Regularization (aSeqDIP and self-guided DIP), re-parameterization (deep decoder, double over-parameterization, and optimal eye surgeon), and early stopping (ES-DIP) methods are considered. The task is denoising with noise level of 0.01. ES-DIP curve stops near iteration 3500 as it is an early stopping method. As observed, regularization-based methods achieve the best results in terms of both robustness to noise over-fitting and reconstruction quality.}}
\label{fig: psnr_curve_compare}
\vspace{-0.4cm}
\end{figure}

\section{Empirical Results} \label{sec: exp DIP comparison}




\subsection{\textcolor{black}{Qualitative Comparison with Data-centric Methods}}



Here, we focus on the four imaging tasks: MRI from undersampled measurements, sparse view CT, in-painting, and non-linear deblurring. For MRI, we utilize the fastMRI dataset\footnote{\tiny{\url{https://github.com/facebookresearch/fastMRI}}}. The multi-coil data is acquired using 15 coils and is cropped to a resolution of 320×320 pixels. To simulate undersampling in the MRI k-space, cartesian masks with 4× acceleration are applied. Additionally, sensitivity maps for the coils are generated using the BART toolbox\footnote{\tiny{\url{https://mrirecon.github.io/bart/}}}. For sparse-view CT image reconstruction, we use the AAPM dataset\footnote{\tiny{\url{https://www.aapm.org/grandchallenge/lowdosect/}}}. The input image with $256 \times 256$ pixels is transformed into its sinogram representation using a Radon transform (the operator $\mathbf{A}$). The forward model assumes a monoenergetic source and no scatter/noise with $y_i = I_0 e^{-[\mathbf{A} \mathbf{x}^*]_i}$, with $I_0$ denoting the number of incident photons per ray (assumed to be $1$ for simplicity) and $i$ indexes the $i$th measurement or detector pixel. We use the post-log measurements for reconstruction, and the sparse-view angles are all equispaced or randomly selected from $180$ equispaced angles. For the tasks of in-painting and non-linear deblurring, we use the CBSD68 dataset\footnote{\tiny{\url{https://github.com/clausmichele/CBSD68-dataset}}}. For nonlinear deblurring, a neural network-approximated forward model is adopted as described in \cite{alkhouriNeuIPS24}. 

In Fig.~\ref{fig: visual}, we present the reconstructed images for the four tasks considered in this study. Each row corresponds to a different task. The first column shows the ground truth (GT) image, while the second column displays the degraded image. Columns 3 onward presents the reconstructed images produced by the data-dependent DM method, and the last three columns show the results obtained using the data-independent methods. 


\textcolor{black}{The results indicate that data-independent DIP methods, such as self-guided DIP and aSeqDIP, can outperform—or match—the performance of data-dependent DM-based methods. This trend is especially evident in MRI and CT reconstruction tasks, where the performance gap is more pronounced. In tasks like inpainting and deblurring, aSeqDIP also shows superior results compared to the DM-based baseline. We hypothesize that the competitiveness of DIP methods arises from the implicit bias of untrained CNNs, which—when overfitting to noise is properly controlled—can yield strong performance without requiring any training datasets. These findings highlight the potential of DIP approaches as viable alternatives to data-intensive methods.}

\subsection{\textcolor{black}{Robustness of DIP-based methods to noise over-fitting}}\label{sec: exp DIP comparison}


Here, we assess the robustness of various DIP methods to noise over-fitting for the denoising task. The average PSNR is computed over 10 images from the ImageNet dataset. All methods were optimized using Adam. aSeqDIP uses a \(10^{-4}\) learning rate and a \(\lambda = 1\) while Self-Guided DIP uses \(3 \times 10^{-4}\) and \(\lambda = 0.1\). ES-DIP used \(10^{-3}\), Deep Decoder method used 0.008, and DOP uses \(10^{-4}\). For OES, sparsity was $5\%$, with a mask training learning rate of \(10^{-2}\) and an image denoising learning rate of \(10^{-3}\). These hyperparameters follow the original papers of each method. As we can see from Fig.~\ref{fig: psnr_curve_compare}, aSeqDIP and Optimal Eye surgeon show the most competitive robustness against noise \textcolor{black}{over-fitting} while aSeqDIP and Self-guided DIP achieve the best PSNR.

\begin{table}[h]
    \centering
    \begin{tabular}{|l|l|l|l|}
        \hline
        \textbf{Method} & \textbf{Year} & \textbf{Category} &  \textbf{Average Run-time (seconds)}   \\
        \hline
        aSeqDIP \cite{alkhouriNeuIPS24} & 2024 & Regularization & 109 $\pm$ 24\\
        \hline
        Self-Guided DIP \cite{liang2024analysis} & 2025 & Regularization & 208 $\pm$ 43 \\
        \hline
        Deep Decoder \cite{heckel2019deep} & 2020 & Network re-parameterization & 144 $\pm$ 48 \\
        \hline
        Double over-parameterization \cite{you2020robust} & 2020 & Network re-parameterization & 156 $\pm$ 29 \\
        \hline
        Optimal Eye Surgeon \cite{ghosh2024optimal} & 2024 & Network re-parameterization & 107 $\pm$ 26 \\
        \hline
        ES-DIP \cite{wang2021early} & 2021 & Early stopping & 78 $\pm$ 45\\
        \hline
    \end{tabular}
    \vspace{0.25cm}
    \caption{\textcolor{black}{Comparison of different DIP methods in terms of average run-time (in seconds)}.}
    \vspace{-0.8cm}
    \label{tab:run-time}
\end{table}

\textcolor{black}{In Table~\ref{tab:run-time}, we present the wall clock run-time required for every method considered in the experiment of this section (using an RTX5000 GPU machine). As observed, ES-DIP reports the lowest run-time as the optimization stops early. These results compare the computational cost of various state-of-the-art methods.}

\section{Open Questions \& Further Directions}

In this section, we discuss open questions and future directions for DIP. Most DIP methods, if not all, utilize Convolution-based architectures. However, hybrid network structures that combine convolutional and attention layers may exhibit similar implicit biases as CNNs. 
Evaluating DIP and its data-less setting within such architectures presents a promising future direction.

Regarding \textcolor{black}{modalities} and tasks, DIP methods have been primarily applied to images and image inverse problems. Extending DIP to non-imaging inverse problems and other data modalities, such as graphs, is another potential avenue worth exploring. For example, in the case of graph data, would a convolution-based implicit prior be sufficient, or would graph neural networks be more suitable?

From a computational perspective, can DIP be made \textit{faster}? DIP is relatively slow compared to the inference time of certain data-centric methods (e.g., the supervised method in MoDL \cite{8434321}). DRP \cite{Li_2023_CVPR} have proposed some efficiency improvements. However, many opportunities for accelerating DIP remain. In particular, recent studies \cite{liang2024analysis,ghosh2024optimal} have explored transferring a pre-trained DIP network to a new image, potentially speeding up reconstruction. This raises a broader question: what does \textit{generalization} mean for DIP? Addressing these questions remains an open research challenge.

Existing theoretical analyses primarily focus on the implicit bias of gradient descent, whereas, in practice, preconditioned methods like ADAM are commonly used to ensure convergence. Moreover, implicit bias toward natural images emerges even with practical learning rates and initialization, rather than the small learning rates and carefully controlled initialization typically assumed in theoretical studies. A promising direction for future research is to extend beyond these simplified settings like NTK and investigate the impact of large step sizes and initialization strategies on the implicit bias of DIP. \textcolor{black}{There are also important questions about the extent to which existing theoretical analyses of DIP explain its strong performance in practice. For example, can the empirical NTK of randomly initialized deep networks like U-Nets act as a useful image filter? Constructing such a filter using real networks and images presents significant computational challenges. Additionally, current analyses are primarily only applicable to the original DIP formulation, which does not perform as well in practice as subsequent formulations. Extending existing theoretical approaches to explain the strong performance of these schemes is another significant research direction.}

An interesting direction is extending DIP to multiple measurements $\mathbf{y}_i$, $i\in \{1,\dots,N\}$ with $N>1$, transitioning from a data-less regime to self-supervised learning. \textcolor{black}{In particular, given $\mathbf{y}_i$, can we use DIP to learn a prior such that at testing time, we are able to improve the performance in terms of acceleration and reconstruction quality? Future works could explore
algorithms for addressing this question, and alongside explore the amount of training data needed and whether the degraded measurements at test-time need to be semantically related.}

\textcolor{black}{While integrating DIP with pre-trained DMs shows promise (e.g., \cite{LianguDigDIP,ddip}), open questions remain. Can DIP be \textit{efficiently} incorporated into accelerated DM samplers for measurement consistency? Addressing this could enable more robust, scalable integrations with generative frameworks.}

\section{Author Biographies and Contact Information}

\noindent \textbf{Ismail Alkhouri} (ismailal@umich.edu,alkhour3@msu.edu) is a Research Scientist at Systems Planning and Analysis (Alexandria, VA 22311), providing technical support to DARPA. He is a research scholar at the University of Michigan (UM) and Michigan State University (MSU). He earned his Ph.D. in Electrical and Computer Engineering from the University of Central Florida in May 2023 and was a postdoctoral researcher at MSU and UM from July 2023 to December 2024. He is a recipient of the 2025 CPAL Rising Stars Award. His research focuses on computational imaging with deep generative models and differentiable methods for combinatorial optimization.

\noindent \textbf{Avrajit Ghosh} (ghoshavr@berkeley.edu) is a Postdoctoral Fellow at the Simons Institute for the Theory of Computing, University of California, Berkeley. He received his Ph.D. degree in Computational Mathematics, Science and Engineering from Michigan State University in 2025. His research focuses on the theoretical foundations of deep learning and optimization.

\noindent \textbf{Evan Bell} (belleva1@msu.edu) is a Ph.D. student in Electrical and Computer Engineering at Johns Hopkins University (Baltimore, MD 21218). From May 2024 to August 2025, he was a post-baccalaureate researcher in the Theoretical Division at Los Alamos National Laboratory and the Department of Compuational Mathematics, Science and Engineering at Michigan State University. He received B.S. degrees in Mathematics and Physics from Michigan State University in 2024. His research focuses on solving inverse problems in computational imaging and physics using deep learning, particularly in limited data settings. 

\noindent  \textbf{Shijun Liang} (liangs16@msu.edu) received his B.S. degree in Biochemistry from the University of California, Davis, CA, USA, in 2017 as well as a Ph.D. degree in the Department of Biomedical Engineering at Michigan State University, East Lansing, MI, USA, in 2025. His research focuses on machine learning and optimization techniques for solving inverse problems in imaging. Specifically, he is interested in machine learning based image reconstruction and in enhancing the robustness of learning-based reconstruction algorithms.


\noindent  \textbf{Rongrong Wang} (wangron6@msu.edu) holds a B.S. in Mathematics and a B.A. in Economics from Peking University and a Ph.D. in Applied Mathematics from the University of Maryland, College Park. She is an Associate Professor at Michigan State University in Computational Mathematics and Mathematics. Previously, she was a postdoctoral researcher at the University of British Columbia. Her research focuses on modeling and optimization for data-driven computation, developing learning algorithms, optimization formulations, and scalable numerical methods with theoretical guarantees. Her work has applications in signal processing, machine learning, and inverse problems.

\noindent  \textbf{Saiprasad Ravishankar} (ravisha3@msu.edu) is an Associate Professor in Computational Mathematics, Science and Engineering, and Biomedical Engineering at Michigan State University. He received the B.Tech. in Electrical Engineering from IIT Madras, India, in 2008, and M.S. and Ph.D. in Electrical and Computer Engineering from UIUC, USA in 2010 and 2014. He was an Adjunct Lecturer and postdoc at UIUC and then postdoc at University of Michigan and Los Alamos National Laboratory (2015–2019). His interests include machine learning, imaging, signal processing, neuroscience, and optimization. He is a member of the IEEE Machine Learning for Signal Processing (MLSP) and Bio Imaging and Signal Processing (BISP) Technical Committees. He received the NSF CAREER Award in 2024.

\ifCLASSOPTIONcaptionsoff
  \newpage
\fi

{\footnotesize
\bibliographystyle{IEEEtranNoDash}
\bibliography{refs}
}

\end{document}